\documentclass[10pt,twocolumn,twoside]{IEEEtran}
\usepackage{amsfonts,dsfont,amssymb,amsbsy,amsmath,paralist,theorem,bm,ifthen,color}
\usepackage[pdfstartview=FitH,bookmarksnumbered,unicode,bookmarksopen=true]{hyperref}
\usepackage{graphicx}
\usepackage{epstopdf}
\usepackage{multirow}
\usepackage{subfigure,relsize}
\usepackage{mathrsfs}
\usepackage{bm}
\usepackage{stfloats}
\usepackage{url}
\usepackage{array}
\usepackage{colortbl,xcolor}
\usepackage{makecell}
\usepackage[noadjust]{cite}
\usepackage{cases,booktabs}
\usepackage{graphicx,algorithmic,algorithm,relsize}
\usepackage[justification=centering]{caption} 

\newtheorem{lem}{Lemma}
\newtheorem{prop}{Proposition}

\newtheorem{rem}{Remark}

\newtheorem{assumption}{Assumption}

\newcommand\diag{\ensuremath{{\rm diag}}}

\graphicspath{{fig/}}

\definecolor{green}{RGB}{34	195	46}
\definecolor{red}{RGB}{220 0 0}

\usepackage{graphicx} 
\allowdisplaybreaks[4]

\title{Environment-Aware Network-Level Design of Pinching-Antenna Systems--Part I: \\Traffic-Aware Case}

\author{
 Yanqing Xu, \IEEEmembership{Senior Member, IEEE,}
         Zhiguo Ding, \IEEEmembership{Fellow, IEEE,}
         Xiu Yin Zhang, \IEEEmembership{Fellow, IEEE,}\\
         Trung Q. Duong, \IEEEmembership{Fellow, IEEE,}
         and Tsung-Hui Chang, \IEEEmembership{Fellow, IEEE}
         \thanks{\smaller[1] Part of this work has been accepted by IEEE GLOBECOM 2026 \cite{xu2026traffic}.}
         \thanks{\smaller[1] Y. Xu is with the School of Science and Engineering, The Chinese University of Hong Kong, Shenzhen, 518172, China (email: xuyanqing@cuhk.edu.cn).}
         \thanks{\smaller[1] Z. Ding is with the School of Electrical \& Electronic Engineering, Nanyang Technological University, 639798, Singapore  (e-mail: zhiguo.ding@ntu.edu.sg).}
         \thanks{\smaller[1] X. Y. Zhang is with the School of Microelectronics, South China University of Technology, Guangzhou 510000, China, and also with the Pazhou Laboratory, Guangzhou 510330, China (e-mail: eexyz@scut.edu.cn).}
         \thanks{\smaller[1] T. Q. Duong is with the Faculty of Engineering and Applied Science, Memorial University, St. John's, NL A1C 5S7, Canada and also with the School of Electronics, Electrical Engineering and Computer Science, Queen's University Belfast, Belfast, U.K. (e-mail: tduong@mun.ca.)}
         \thanks{\smaller[1] T.-H. Chang is with the School of Artificial Intelligence, The Chinese University of Hong Kong, Shenzhen, 518172, China (email: changtsunghui@cuhk.edu.cn).} 
        }

\date{\today}

\begin{document}

\maketitle

\begin{abstract}
    Existing studies on pinching-antenna systems have primarily focused on link-level design, where system parameters are optimized for a given user set according to user-specific communication metrics. Such designs provide an effective means of exploiting the spatial reconfigurability of pinching antennas for short-term user-centric transmission. Complementing this perspective, this two-part paper develops an environment-aware network-level design framework, where pinching-antenna configurations are optimized according to long-term spatial information and evaluated through network-level performance metrics. Part I focuses on the traffic-aware case, where user presence is modeled statistically by a spatial traffic map and performance is optimized and evaluated in a traffic-aware sense; Part II addresses the geometry-aware case in obstacle-rich environments by explicitly modeling line-of-sight blocking and optimizing region-wide robustness objectives. In Part~I, we introduce traffic-weighted average signal-to-noise ratio (SNR) metrics and formulate two traffic-aware deployment problems: (i) maximizing the traffic-weighted network
    average SNR, and (ii) a fairness-oriented traffic-restricted max--min
    average-SNR design over traffic-dominant grids. To solve these nonconvex
    problems with low complexity, we reveal and exploit their separable structures.
    For the network-average objective, we establish unimodality properties of the hotspot-induced components and develop a candidate-based low-complexity
    method that only needs to evaluate the objective at a small set of candidate
    antenna positions. For the traffic-restricted max--min objective, we develop a
    block coordinate descent framework where each coordinate update reduces to a 
    one-dimensional subproblem via an epigraph reformulation and bisection.
    Simulations show that traffic-aware pinching-antenna positioning consistently
    outperforms representative fixed and heuristic traffic-aware deployments in the
    considered setups.
\end{abstract}


\section{Introduction} 

In addition to link-level transmission performance, modern wireless systems are increasingly assessed by area-wide service quality. In practice, operators must provide reliable and fair coverage over an entire service region, especially in deployments where both user demand and propagation conditions vary remarkably across space. User presence is often non-uniform and clustered around traffic hotspots \cite{niu2011tango}, while blockage, scattering, and irregular layouts create location-dependent channels with LoS corridors and shadowed zones \cite{zeng2021toward,bi2019engineering,xu2025distributed}. Consequently, network performance is jointly governed by the spatial traffic distribution and geometry-induced propagation characteristics. This motivates a complementary network-level design perspective, where performance is characterized over the service region using long-term spatial information.

Conventional network-level optimization typically tunes configuration parameters of a fixed infrastructure, e.g., antenna tilts/beam patterns, transmit powers, and cell-level control knobs, to improve region-wide coverage and capacity. A representative example is the self-organizing-network line of work on coverage-and-capacity optimization, where remote electrical tilts and power settings are adapted based on network measurements (e.g., call traces) to address overshooting and coverage holes while balancing spectral efficiency; related formulations also optimize antenna tilts via utility/fairness criteria at the network level \cite{partov2014utility,buenestado2016self}.
In parallel, map-based paradigms such as radio environment maps (REMs) and channel-knowledge maps (CKMs) advocate learning location-tagged propagation statistics to enable environment-aware resource management across grids \cite{zeng2024tutorial}. However, CKM-centric designs primarily capture propagation knowledge and do not explicitly encode where/when users appear and demand concentrates; as emphasized by the perception-embedding-map framework (PEMNet), network optimization fundamentally benefits from jointly embedding channel and fine-grained spatial–temporal traffic knowledge, since traffic maps determine the relevant operating points and bottlenecks that channel maps alone cannot reveal \cite{li2026pemnet}. Recent data-driven ``digital-twin'' frameworks further aim to support network-level optimization by building measurement-grounded simulators and localized channel/coverage maps so that network behaviors can be evaluated and optimized off-line rather than through costly trial-and-error in live networks \cite{luo2023srcon,peng2025rf}.

Despite their success, these approaches are fundamentally built on a fixed infrastructure and therefore share an inherent limitation: the radiating sites remain physically anchored, so optimization can primarily reshape \emph{how} a given site radiates (e.g., via tilt, power, or beam pattern adaptation) but cannot change \emph{where} energy is injected into the environment. Consequently, both environment-induced and demand-induced non-uniformities can persist at the network level. In blockage-rich environments, coverage tails and worst-location performance are often dictated by a small set of severely shadowed or far-corner grids that are difficult to lift through parameter tuning alone. Meanwhile, when user traffic is highly heterogeneous, fixed-site configurations may struggle to deliver sufficient service to shifting hotspots without sacrificing performance elsewhere. 

These limitations motivate architectures that introduce an additional,
structured spatial degree of freedom at the infrastructure level.
Several flexible-antenna technologies have recently been investigated
toward this goal, including movable antennas, six-dimensional movable
antennas, and fluid antenna systems, which enable the antenna position,
orientation, or effective radiation port to be reconfigured
\cite{wong2020fluid,
zhu2025tutorial,shao20246d,shao20256d}. Such approaches provide valuable spatial
flexibility, but the achievable reconfiguration is typically confined to
the physical region occupied by the antenna structure.
Pinching-antenna systems offer a complementary form of spatial
reconfigurability by enabling on-demand signal radiation at selectable
locations along an extended guiding medium, e.g., dielectric waveguide
\cite{suzuki2022pinching}, allowing the effective radiation origin to be
reconfigured over a potentially much larger spatial range. 
Specifically, the radio frequency (RF) signal is guided along a medium and selectively radiated at configurable points, creating short-hop links near users and enabling the coverage ``injection points'' to be reconfigured after deployment. This expands the design space beyond conventional controls that only reshape how a fixed site radiates (e.g., tilt, power, beamforming), by additionally controlling where energy is injected via repositioning/activating radiating points along the medium. Such spatial actuation is particularly valuable for area-wide objectives, as it can steer coverage toward high-demand regions and mitigate geometry-induced dead zones (e.g., obstacle shadows) by choosing more favorable radiating locations.
In this two-part paper, we develop an environment-aware, network-level design framework for pinching-antenna systems that accounts for both demand distribution and propagation geometry, together with low-complexity optimization algorithms to configure radiating points for improved region-wide performance.

In the literature, pinching-antenna systems have received rapidly growing attention due to their unique capability of enabling large-scale spatially reconfigurable transmission by creating configurable radiating points along an extended guiding medium. For example, the work in \cite{ding2024flexible} established fundamental electromagnetic/system models and showcased that pinching antennas can flexibly form strong LoS-dominant links with reduced propagation loss by adjusting the radiation locations. Building on this foundation, the downlink rate maximization problem in a single-waveguide setting was investigated in \cite{xu2025rate}, where a two-stage design was proposed to jointly reduce large-scale attenuation and promote favorable signal combining at the receiver. The study was subsequently extended to multiuser and multi-waveguide regimes in \cite{bereyhi2025mimo,zhang2025two,papanikolaou2025physical,zhou2025gradient}, demonstrating that jointly optimizing the pinching-antenna locations and transmit beamforming can substantially improve spectral and/or energy efficiencies. To further enhance spectral efficiency through user multiplexing, NOMA-enabled pinching-antenna transmission was explored in \cite{wang2025antenna,vashakidze2025joint,xue2025resource}, where multiple users' signals are superposed and conveyed through the same guiding medium. In parallel, the integration of pinching antennas with emerging paradigms such as integrated sensing and communication (ISAC) has been studied in \cite{ding2025pinching,khalili2025pinching,mao2025multi}. By leveraging the capability of pinching antennas to reconfigure LoS links, an environment division multiple access (EDMA) technique was proposed in \cite{ding2025EDMA}, further highlighting the potential of pinching antennas for multiuser communications. 
 Recent studies have also investigated performance analysis and practical
two-state pinching-antenna architectures, e.g.,
\cite{tyrovolas2025performance,tyrovolas2026many,galanopoulou2026viterbi}. These works mainly
characterize physical-layer performance or discrete implementations.
Beyond dielectric-waveguide realizations, which are particularly attractive
at high carrier frequencies, the generalized pinching-antenna concept can
also be implemented over other guiding structures \cite{xu2025generalized}.
For instance, \cite{wang2025generalized} developed an LCX-based architecture
in which radiation locations are reconfigured by selectively activating or
deactivating leakage slots, thereby transforming conventional fixed-leakage
LCX into an on-demand radiating structure while retaining the robustness
and cost advantages of coaxial cables. More recently,
\cite{xu2026generalized-rs} investigated a radio-stripe-based realization,
where distributed active antenna processing units can be selectively
activated to form reconfigurable radiation/reception points, thereby
introducing additional flexibility in configuring the effective wireless
access topology.

Existing design-oriented studies have primarily investigated pinching-antenna systems from a link-level perspective, where the antenna locations and transmission parameters are optimized for a given set of active users according to user-level metrics such as SINR and achievable rate. These studies have established important design principles for exploiting the spatial reconfigurability of pinching antennas and form the foundation for their user-centric operation. Building on these advances, a complementary network-level question arises: how should the reconfigurable radiation locations be configured according to long-term spatial information of the service region, before the instantaneous user realization is known? Addressing this question requires performance metrics and optimization methodologies defined over spatial user distributions or environment maps rather than a particular user realization. This motivates the environment-aware network-level framework developed in this two-part paper.

Building on these observations, this two-part paper develops an environment-aware network-level design framework for pinching-antenna systems under two complementary settings: the \emph{traffic-aware} case and the \emph{geometry-aware} case. 
Different from conventional environment-aware network optimization over
fixed infrastructure, our framework treats the reconfigurable radiation
locations as network-level design variables and configures them according
to long-term spatial information of the service region.
In Part~I, we consider traffic-aware network-level optimization,
where user presence is characterized statistically by a spatial
traffic distribution and the radiation locations are optimized accordingly
to improve traffic-aware performance over the service region.
Part~II complements this study by focusing on geometry-aware design in obstacle-rich environments, where the environment-induced LoS blocking structure is explicitly modeled and area-wide objectives such as coverage and worst-location robustness are optimized.
In Part~I, the main contributions are summarized as follows:
\begin{itemize}
    \item {\bf Traffic-aware network-level modeling and metrics:}
    We develop a tractable traffic-aware network-level framework for
    pinching-antenna systems, where the reconfigurable radiation locations are
    configured according to long-term spatial traffic demand. Specifically,
    we (i) model user presence over the service region via a Gaussian-mixture
    hotspot traffic map, and (ii) define a traffic-weighted network average-SNR
    metric that evaluates long-term coverage quality under random LoS/non LoS (NLoS) fading. To enable efficient optimization, we further introduce a grid-based approximation that discretizes the continuous traffic integral into a numerically tractable weighted sum over representative spatial grids.
    \item \textbf{Traffic-aware problem formulations and optimizations:}
    Building on the proposed traffic-aware grid model, we formulate two
    pinching-antenna deployment problems: (i) maximizing the traffic-weighted
    network average SNR, and (ii) a fairness-oriented traffic-restricted max--min
    average-SNR design that maximizes the worst average SNR over an active-grid set.
    Beyond the problem formulations, we reveal and exploit their
    problem-specific structures to enable low-complexity optimization.
    Specifically, for the network-average objective, we express each per-waveguide
    term as a traffic-weighted superposition of hotspot-induced unimodal components,
    reveal the ensuing merge--split behavior,
    and develop a candidate-based low-complexity algorithm that only needs to
    compare the objective values over a small set of candidate antenna positions.
    For the traffic-restricted max--min objective, we adopt a block coordinate
    descent (BCD) framework and show that each coordinate update reduces to a
    globally solvable one-dimensional subproblem, which can be efficiently handled
    via bisection on its epigraph form.
\end{itemize}
Extensive simulations demonstrate that the proposed low-complexity
traffic-aware algorithms achieve compelling performance in terms of the considered metrics, and that pinching-antenna systems significantly outperform the
fixed-antenna benchmark across a wide range of system settings.

The rest of the paper is organized as follows. Section \ref{sec:system_traffic_channel} presents the system model, the probabilistic LoS/NLoS channel model, and the proposed traffic-aware network-level performance metrics together with their grid-based approximation. Section \ref{sec:avg_snr_opt} formulates the traffic-weighted network average-SNR maximization problem and develops an efficient candidate-based low-complexity algorithm. Section \ref{sec:traffic_restricted_maxmin} considers a fairness-oriented traffic-restricted max–min average-SNR design and proposes a low-complexity BCD algorithm with bisection-based coordinate updates. Section \ref{sec:simulation}  provides numerical results and performance discussions. Finally, Section \ref{sec: conclusion} concludes the paper.


\begin{figure}[!t]
	\centering
	\includegraphics[width=0.92\linewidth]{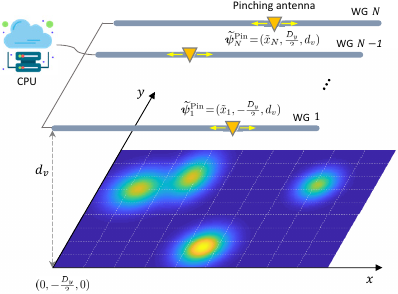}\\
        \captionsetup{justification=justified, singlelinecheck=false, font=small}	
        \caption{Traffic-aware pinching-antenna systems.} \label{fig: system model} 
\end{figure}

\section{Models and Performance Metrics}
\label{sec:system_traffic_channel}
In this section, we present the system, channel, and traffic models, as well as the performance metrics and design objectives considered in this work.

\subsection{System Model}
\label{subsec:system_model}
Let us first introduce the system model. In this work, we consider a downlink pinching-antenna system where $N$ dielectric
waveguides are deployed above a rectangular communication region of size
$D_x \times D_y$, as shown in Fig. \ref{fig: system model}.
The waveguides are aligned along the $x$-axis at a height $d_v$ and
arranged uniformly in the vertical ($y$-) direction across the coverage
region.
Let $\mathcal{N} \triangleq \{1,2,\ldots,N\}$ denote the index set of
waveguides. The spacing between adjacent waveguides is denoted by $d_h$ and is
chosen as $d_h = \frac{D_y}{N-1} \gg \lambda$, where $\lambda$ is the carrier wavelength.
This ensures that mutual coupling between different waveguides can be
neglected. The coordinates of the feed point on the $n$-th waveguide are $\boldsymbol{\psi}_{0,n} = \big[0,\ (n-1)d_h - \tfrac{D_y}{2},\ d_v\big]^{\mathsf T}, \forall n \in \mathcal{N}$.
Each waveguide is equipped with a single pinching antenna.
The position of the pinching antenna on the $n$-th waveguide is denoted
by $\boldsymbol{\psi}^{\mathrm{pin}}_n = \big[\widetilde{x}_n, \widetilde y_n, d_v\big]^{\mathsf T}, \widetilde{x}_n \in [0,D_x]$ is tunable and $\widetilde y_n = (n-1)d_h - \tfrac{D_y}{2}$ is fixed. The pinching antennas' horizontal positions are collected into a vector $\widetilde{\boldsymbol{x}} = [\widetilde{x}_1,\ldots,\widetilde{x}_N]^{\mathsf T}$.
All waveguides are connected to a common baseband processing unit,
which enables coherent joint transmission.

In this work, we focus on the downlink SNR experienced by a typical scheduled user.
In each time–frequency resource, at most one user is served, and the BS transmits one data symbol $s$ with
$\mathbb{E}[|s|^2]=1$ across all $N$ pinching antennas with total power $P$.
For a user located at position $\boldsymbol{\psi} = [x,y,0]^{\mathsf T}, 0 \le x \le D_x, -\tfrac{D_y}{2} \le y \le \tfrac{D_y}{2}$, 
let $\boldsymbol{h}(\boldsymbol{\psi},\widetilde{\boldsymbol{x}}) = \big[  h_1(\boldsymbol{\psi},\widetilde{\boldsymbol{x}}),\ldots,  h_N(\boldsymbol{\psi},\widetilde{\boldsymbol{x}}) \big]^{\mathsf T}$ denote the vector of complex baseband channel coefficients from the $N$ pinching antennas to this user.

The received signal is given by
\begin{align}
    y
    = \sqrt{P}\,
      \boldsymbol{h}^{\mathsf T}(\boldsymbol{\psi},\widetilde{\boldsymbol{x}})
      \boldsymbol{w}\, s + n,
\end{align}
where $\boldsymbol{w}\in\mathbb{C}^{N}$ is the beamforming vector, and
$n\sim\mathcal{CN}(0,\sigma^2)$ denotes additive white Gaussian noise (AWGN),
with $\mathcal{CN}(0,\sigma^2)$ denoting the circularly symmetric complex
Gaussian distribution with zero mean and variance $\sigma^2$. 
Under single-user operation, the BS applies maximum-ratio transmission
(MRT) and the MRT beamformer is given by 
\begin{align}
    \boldsymbol{w}
    = \frac{
        \boldsymbol{h}^\ast(\boldsymbol{\psi},\widetilde{\boldsymbol{x}})
      }{
        \big\|\boldsymbol{h}(\boldsymbol{\psi},\widetilde{\boldsymbol{x}})\big\|_2
      }, \ \|\boldsymbol{w}\|^2 = 1,
\end{align}
which is optimal in terms of SNR for a single scheduled user under a
sum-power constraint. Therefore, the corresponding instantaneous SNR at position $\boldsymbol{\psi}$ is given by
\begin{align}
    \Gamma(\boldsymbol{\psi},\widetilde{\boldsymbol{x}})
    = \frac{P}{\sigma^2}
      \big\|\boldsymbol{h}(\boldsymbol{\psi},\widetilde{\boldsymbol{x}})\big\|_2^2
    \triangleq \rho\,
      \big\|\boldsymbol{h}(\boldsymbol{\psi},\widetilde{\boldsymbol{x}})\big\|_2^2,
\end{align}
where $\rho \triangleq P/\sigma^2$ denotes the transmit SNR.

\subsection{Channel Model with Random LoS and NLoS Components}
\label{subsec:channel_model}

The wireless channel between the $n$-th pinching antenna and a user at position $\boldsymbol{\psi} = [x,y,0]^{\mathsf T}$ is modeled as the superposition of a probabilistic LoS component and a  NLoS scattering component \cite{xu2025pinching-nlos}, which is given by
\begin{align}
    h_n(\boldsymbol{\psi},\widetilde{\boldsymbol{x}})
    = \xi_n\, h^{\mathrm{LoS}}_n(\boldsymbol{\psi},\widetilde{\boldsymbol{x}})
      + h^{\mathrm{NLoS}}_n(\boldsymbol{\psi},\widetilde{\boldsymbol{x}}).
\end{align}

For a user at $\boldsymbol{\psi} = [x,y,0]^{\mathsf T}$, the LoS channel component $h^{\mathrm{LoS}}_n(\boldsymbol{\psi},\widetilde{\boldsymbol{x}})$ is given by
\footnote{For analytical tractability, the main channel model neglects in-waveguide attenuation, allowing us to focus on the network-level traffic-aware design under a clean analytical framework. The impact of practical in-waveguide attenuation is quantitatively evaluated in Sec.~\ref{sec:simulation} through numerical simulations.
}
\begin{align}
    h^{\mathrm{LoS}}_n(\boldsymbol{\psi},\widetilde{\boldsymbol{x}})
    = \frac{\sqrt{\eta}\,
      e^{-j\left(\frac{2\pi}{\lambda} \sqrt{(x - \widetilde x_n)^2 + C_n}
                   + \frac{2\pi}{\lambda_g}\widetilde{x}_n \right)}}
           {\sqrt{(x - \widetilde x_n)^2 + C_n}},
\end{align}
where $C_n = \big(y - ((n-1)d_h - \tfrac{D_y}{2})\big)^2 + d_v^2$, $\lambda$ is the free-space wavelength, $\lambda_g$ is the guided wavelength in the waveguide, and $\eta = \frac{c^2}{(4\pi f_c)^2}$ is a frequency-dependent constant, with $c$ denoting the speed of light and $f_c$ denoting the carrier frequency. We also define the distance between the $n$-th pinching antenna and the user as $r_n(\boldsymbol{\psi},\widetilde{\boldsymbol{x}}) \triangleq \sqrt{(x - \widetilde{x}_n)^2 + C_n}$.

The Bernoulli random variable $\xi_n \in \{0,1\}$ indicates the presence ($\xi_n = 1$) or absence ($\xi_n = 0$) of a direct LoS path between the user and the antenna. The probability of a LoS path is a function of the distance $r_n(\boldsymbol{\psi},\widetilde{\boldsymbol{x}})$, which is modeled as
\begin{align}
    \Pr[\xi_n = 1] = p_{\mathrm{LoS}}\left(r_n(\boldsymbol{\psi},\widetilde{\boldsymbol{x}})\right) = e^{-\beta\, r_n^2(\boldsymbol{\psi},\widetilde{\boldsymbol{x}})},
\end{align}
where $\beta \geq 0$ is an environment-dependent parameter that controls the likelihood of a LoS path \cite{3gpp2020channel,bai2014analysis,ding2025blockage}.

The NLoS channel component, $h^{\mathrm{NLoS}}_n(\boldsymbol{\psi},\widetilde{\boldsymbol{x}})$, is modeled by the channel gain as a sum of multipath components from scattered paths. We model the NLoS channel as a sum of scattered multipath components, given by
\begin{align}
    h^{\mathrm{NLoS}}_n(\boldsymbol{\psi},\widetilde{\boldsymbol{x}})
    = \sum_{c=1}^{N_c} g_{n,c}(\boldsymbol{\psi},\widetilde{\boldsymbol{x}}),
\end{align}
where $N_c$ is the number of NLoS clusters, and $g_{n,c}(\boldsymbol{\psi},\widetilde{\boldsymbol{x}})$ represents the small-scale fading of the $c$-th NLoS cluster for the $n$-th antenna, which is assumed to be a zero-mean complex Gaussian random variable with a user-to-antenna distance-dependent variance 
\begin{align}
    g_{n,c}(\boldsymbol{\psi},\widetilde{\boldsymbol{x}}) \sim \mathcal{CN}\left( 0, \frac{\mu_{n,c}^2}{r_n^2(\boldsymbol{\psi},\widetilde{\boldsymbol{x}})} \right),
\end{align}
where $\mu_{n,c}^2$ represents the average power of the $c$-th NLoS cluster for the $n$-th antenna. Assume that the NLoS paths are independent, the NLoS path gain follows $h^{\mathrm{NLoS}}_n(\boldsymbol{\psi},\widetilde{\boldsymbol{x}}) \sim \mathcal{CN}(0,\frac{\mu_{n}^2}{r_n^2(\boldsymbol{\psi},\widetilde{\boldsymbol{x}})})$, where $\mu_{n}^2 = \sum_{c=1}^{N_c}\mu_{n,c}^2$.

We note that the above probabilistic LoS/NLoS model is adopted to capture random blockage effects while maintaining a tractable characterization of the average channel power. It is particularly suitable when detailed environment geometry is unavailable and the LoS blockage can only be characterized statistically. When such geometry information is available, the LoS state can instead be determined explicitly from the environment map, as considered in Part II. More general propagation models, such as Rician fading with distance- or environment-dependent parameters, can also be considered when finer-grained channel information is available.

In this work, we do not focus on the instantaneous SNR of a
particular user, but network-level performance metrics, i.e., the average SNR according to where users are likely to appear.
To this end, we next introduce a traffic-aware spatial user model based
on Gaussian hotspots, which characterizes the distribution of active
user positions and serves as the basis for the subsequent grid-based
approximation and optimization of $\widetilde{\boldsymbol{x}}$.

\subsection{Traffic Model}
\label{subsec:traffic_model}

Let $\boldsymbol{U}=[X,Y]^{\mathsf T}$ denote the random horizontal position
of a typical scheduled user on the plane $z=0$.
We model the probability of user appearance over a region as a mixture
of two-dimensional Gaussian ``hotspots''\footnote{ The Gaussian-mixture
model is adopted as a tractable and
interpretable representation of spatially clustered traffic. In practice,
traffic distributions may be non-Gaussian or time-varying and can be
obtained from historical measurements or data-driven prediction methods.
Accordingly, the assumed traffic density can be replaced by an estimated
spatial traffic distribution when such information is available. }.
According to \cite{zhang2021predictive}, the probability density function (PDF) of $\boldsymbol{U}$ is given by
\begin{align}
  f_{\boldsymbol{U}}(x,y)
  = \sum_{\ell=1}^{L} \alpha_{\ell}\,
    \mathcal{N}\!\big([x,y]^{\mathsf T};
        \boldsymbol{\mu}_{\ell},\boldsymbol{\Sigma}_{\ell}\big),
  \label{eq:pdf_hotspots}
\end{align}
where $L$ is the number of hotspots, $\alpha_{\ell}\ge 0$ denotes the traffic
intensity (mixture weight) of hotspot $\ell$ satisfying
$\sum_{\ell=1}^{L}\alpha_{\ell}=1$, and
$\mathcal{N}(\cdot;\boldsymbol{\mu}_{\ell},\boldsymbol{\Sigma}_{\ell})$ denotes
the Gaussian probability density function (PDF) with mean $\boldsymbol{\mu}_{\ell}$ and covariance
$\boldsymbol{\Sigma}_{\ell}$.
Here,
$\boldsymbol{\mu}_{\ell}=[\mu_{\ell,x},\mu_{\ell,y}]^{\mathsf T}$ specifies the
hotspot center on the horizontal plane, i.e., $\mu_{\ell,x}$ and $\mu_{\ell,y}$
are the $x$- and $y$-coordinates of hotspot $\ell$, respectively.
Moreover, we adopt a diagonal covariance model
$\boldsymbol{\Sigma}_{\ell}=\diag(\sigma_{\ell,x}^2,\sigma_{\ell,y}^2)$, where
$\sigma_{\ell,x}^2$ and $\sigma_{\ell,y}^2$ characterize the traffic spread
(variance) of hotspot $\ell$ along the $x$- and $y$-directions, respectively.

At each scheduling instant, the position of the scheduled user can be
regarded as an independent realization of $\boldsymbol{U}$.
This traffic-aware model captures the fact that users are much more
likely to appear around a few dominant hotspots (such as offices,
meeting rooms, or production lines), while user density decays away from these centers according to the
Gaussian tails.

\subsection{Traffic-Aware Network Performance Metric}
\label{subsec:perf_metrics}

In fading environments with probabilistic LoS/NLoS conditions, it is
important to evaluate how well a given pinching-antenna deployment
matches the spatially non-uniform traffic demand. In this work, we adopt
a traffic-aware metric based on the average SNR, which captures the long-term link quality experienced by a typical active user whose location follows the hotspot traffic map.

For a fixed deployment $\widetilde{\boldsymbol{x}}$, the local average
SNR at position $(x,y)$ is defined as
\begin{align}
  \bar{\Gamma}(x,y;\widetilde{\boldsymbol{x}})
  \triangleq
    \mathbb{E}\big[
      \Gamma(\boldsymbol{\psi},\widetilde{\boldsymbol{x}})
      \,\big|\,
      \boldsymbol{\psi} = [x,y,0]^{\mathsf T}
    \big],
\end{align}
where the expectation is taken over random LoS/NLoS channels for a user
fixed at $(x,y)$.


Since the scheduled-user location $\mathbf{U}$ is random with spatial
distribution $f_U(x,y)$, a natural network-level metric is the expected
average SNR experienced by a typical scheduled user. Accordingly, we
define the network average SNR as
\begin{align}
  \bar{\Gamma}_{\mathrm{net}}(\widetilde{\boldsymbol{x}})
  \triangleq
    \iint \bar{\Gamma}(x,y;\widetilde{\boldsymbol{x}})
         f_{\boldsymbol U}(x,y)\,
         \mathrm{d}x\,\mathrm{d}y.
  \label{eq:net_avg_SNR_cont}
\end{align}
Since $f_U(x,y)$ is the probability density of the scheduled-user
location, $\bar{\Gamma}_{\rm net}(\tilde{\mathbf{x}})$ represents the
long-term average SNR experienced by a typical scheduled user over
different spatial locations and channel realizations. Locations with 
higher traffic probability contribute more to this average because 
a scheduled user is more likely to appear there.

As seen, directly optimizing \eqref{eq:net_avg_SNR_cont} is
non-trivial, since it involves a two-dimensional integral over
$f_{\boldsymbol U}(x,y)$ and an expectation over the hybrid LoS/NLoS
channel, neither of which admits a simple closed-form expression in
general. This motivates the grid-based approximation introduced in the
next subsection, where the service area and traffic distribution are
discretized into a finite set of representative grids to enable
tractable optimization of $\widetilde{\boldsymbol{x}}$.

\subsection{Grid-Based Approximation}
\label{subsec:grid_metrics}

To obtain a tractable network-level formulation, we discretize the
rectangular communication region of size $D_x \times D_y$ into
$N_h$ grids along the horizontal ($x$-) direction and $N_v$ grids along
the vertical ($y$-) direction.
Let
\begin{align}
    \Delta_u \triangleq \frac{D_x}{N_h}, \qquad
    \Delta_v \triangleq \frac{D_y}{N_v}
\end{align}
denote the grid sizes in the horizontal and vertical directions,
respectively.
The center of grid $(u,v)$ is then given by
\begin{align}
    x_u = \Big(u-\tfrac{1}{2}\Big)\Delta_u, \qquad
    y_v = -\tfrac{D_y}{2} + \Big(v-\tfrac{1}{2}\Big)\Delta_v,
\end{align}
for $u = 1,\ldots,N_h$ and $v = 1,\ldots,N_v$.
The corresponding representative user position in grid $(u,v)$ on the
plane $z=0$ is $\boldsymbol{\psi}_{u,v} = [x_u,y_v,0]^{\mathsf T}.$

Using this grid-based approximation, the continuous random position
$\boldsymbol{U}$ is approximated by a discrete random variable that
takes values in the finite set of grid centers
$\{\boldsymbol{\psi}_{u,v}\}$ with probabilities $\{p_{u,v}\}$.
In particular, based on the Gaussian-mixture traffic model in
\eqref{eq:pdf_hotspots}, the probability of grid $(u,v)$ can be
approximated by
\begin{align}
  p_{u,v} \approx f_{\boldsymbol U}(x_u,y_v)\,\Delta_u\Delta_v,
\end{align}
followed by a normalization such that $\sum_{u=1}^{N_h}\sum_{v=1}^{N_v} p_{u,v} = 1$.
In this way, expectations with respect to $\boldsymbol{U}$ can be
approximated by finite weighted sums over the grids.

For a given pinching-antenna deployment $\widetilde{\boldsymbol{x}}$,
the instantaneous SNR experienced by a user located in grid $(u,v)$ is
approximated by
\begin{align}
    \Gamma_{u,v}(\widetilde{\boldsymbol{x}})
    \triangleq
      \Gamma(\boldsymbol{\psi}_{u,v},\widetilde{\boldsymbol{x}}),
\end{align}
with corresponding average SNR defined as
\begin{align}
    \bar{\Gamma}_{u,v}(\widetilde{\boldsymbol{x}}) \triangleq
      \mathbb{E}\big[
        \Gamma_{u,v}(\widetilde{\boldsymbol{x}})
      \big],
\end{align}
where the expectation and probability are taken over the LoS/NLoS
channels, conditioned on the user being in
grid $(u,v)$.

Under this grid-based approximation, the traffic-aware
average-SNR-based metric in \eqref{eq:net_avg_SNR_cont} is evaluated
by
\begin{align}
  \bar{\Gamma}_{\mathrm{net}}(\widetilde{\boldsymbol{x}})
  \approx
    \sum_{u=1}^{N_h}\sum_{v=1}^{N_v}
      p_{u,v}\,
      \bar{\Gamma}_{u,v}(\widetilde{\boldsymbol{x}}),
      \label{eq:net_avg_SNR_grid} 
\end{align}
which provides numerically tractable discrete counterparts of the
continuous-domain metrics \eqref{eq:net_avg_SNR_cont}.
This grid-based expression forms the basis for the optimization
problems studied in the subsequent sections.

\section{Traffic-Weighted Average-SNR Optimization}
\label{sec:avg_snr_opt}

In this section, we aim to maximize the traffic-weighted network average SNR by optimizing the pinching-antenna positions for the given user traffic distribution. In what follows, based on the grid-based approximation, we first derive closed-form expressions for the per-grid and network average SNR. We then formulate the corresponding optimization problem and develop an efficient algorithm to solve it.

\subsection{Per-Grid and Network Average SNR Derivations}
\label{subsec:avg_snr_grid}

Recall that, under the grid-based model, the communication region is
discretized into grids indexed by $(u,v)$ with centers
$\boldsymbol{\psi}_{u,v} = [x_u,y_v,0]^{\mathsf T}$, and the
traffic weight of grid $(u,v)$ is $p_{u,v}$.
For a given deployment $\widetilde{\boldsymbol{x}}$, the instantaneous
SNR experienced by a user located in grid $(u,v)$ is denoted by
\begin{align}
  \Gamma_{u,v}(\widetilde{\boldsymbol{x}}) = \rho\, \big\| \boldsymbol{h}(\boldsymbol{\psi}_{u,v},\widetilde{\boldsymbol{x}}) \big\|_2^2 = \rho \sum_{n=1}^{N} \big| h_n(\boldsymbol{\psi}_{u,v},\widetilde{\boldsymbol{x}}) \big|^2,
\end{align}
where $\boldsymbol{h}(\boldsymbol{\psi}_{u,v},\widetilde{\boldsymbol{x}})$
collects the channels from all $N$ pinching antennas to grid $(u,v)$, and the channels of antennas are independent.

\begin{lem}
\label{lem:average_snr}
For a given pinching-antenna deployment $\widetilde{\boldsymbol{x}}$, the
local average SNR in grid $(u,v)$ admits the closed-form expression
\begin{align}
  \label{eq:avg_snr_uv_closed_form}
    \bar{\Gamma}_{u,v}(\widetilde{\boldsymbol{x}})
    = \rho \sum_{n=1}^{N}
       \frac{ \eta\, \exp\!\big(-\beta r_{n,u,v}^2(\widetilde{x}_n)\big) + \mu_n^2 }
       { r_{n,u,v}^2(\widetilde{x}_n) },
\end{align}
where $r_{n,u,v}(\widetilde{x}_n) \triangleq \sqrt{(x_u-\widetilde{x}_n)^2+C_{n,v}} $
denotes the distance between pinching antenna $n$ and grid $(u,v)$, with 
$C_{n,v} \triangleq (y_v-((n-1)d_h - \tfrac{D_y}{2}))^2+d_v^2$.
With the traffic weights $\{p_{u,v}\}$, the corresponding
traffic-aware network average SNR is given by
\begin{align}
    \bar{\Gamma}_{\mathrm{net}}(\widetilde{\boldsymbol{x}})
    \approx
      \rho \sum_{u=1}^{N_h}\sum_{v=1}^{N_v}
        p_{u,v}\, \sum_{n=1}^{N}
       \frac{ \eta\, \exp\!\big(-\beta r_{n,u,v}^2(\widetilde{x}_n)\big) + \mu_n^2 }
       { r_{n,u,v}^2(\widetilde{x}_n) }.
    \label{eq:net_avg_SNR_grid_final}
\end{align}
\end{lem}

{\emph{Proof:}} See Appendix \ref{appd: average snr}.  \hfill$\blacksquare$

\subsection{Problem Formulation}
\label{subsec:avg_snr_problem}

Based on \eqref{eq:net_avg_SNR_grid_final}, the traffic-weighted network average-SNR maximization problem is formulated as
\begin{subequations} \label{prob:P_Ave}
\begin{align}
\max_{\widetilde{\boldsymbol{x}}} \quad
& \bar{\Gamma}_{\mathrm{net}}(\widetilde{\boldsymbol{x}})
\label{prob:P_Ave_obj}\\
\text{s.t.}\quad
& 0 \le \widetilde{x}_n \le D_x,\quad \forall n \in \mathcal{N},
\label{prob:P_Ave_box}
\end{align}
\end{subequations}
where $\widetilde{\boldsymbol{x}}=[\widetilde{x}_1,\ldots,\widetilde{x}_N]^{\mathsf T}$ collects the pinching-antenna positions along the $N$ waveguides.
 
Problem \eqref{prob:P_Ave} has a simple box-constrained feasible set and a smooth objective:
$\bar{\Gamma}_{\mathrm{net}}(\widetilde{\boldsymbol{x}})$ is continuously differentiable in $\widetilde{\boldsymbol{x}}$
since each per-grid average SNR depends smoothly on the distances $r_{n,u,v}(\widetilde{x}_n)$.
However, \eqref{prob:P_Ave} is generally nonconvex because the average SNR term in
\eqref{eq:avg_snr_uv_closed_form} involves nonlinear distance-dependent components (e.g., the path-loss factor
$1/r_{n,u,v}^2(\widetilde{x}_n)$ and the distance-dependent LoS-probability factor embedded in the numerator), which can lead to multiple stationary points. 
Despite the nonconvexity of \eqref{prob:P_Ave}, the objective admits a favorable
separable structure across pinching antennas, which will be formalized in the
following lemma and will be exploited to develop low-complexity algorithms.

\begin{lem}\label{lem:net_traffic_problem_transform}
The traffic-weighted network average SNR in \eqref{prob:P_Ave} can be decoupled over the pinching antennas as
\begin{align}
    \bar{\Gamma}_{\mathrm{net}}(\widetilde{\boldsymbol{x}}) = \sum_{n=1}^{N} f_n(\widetilde{x}_n),
\end{align}
where the $n$-th component $f_n(\cdot)$ is given by
\begin{align}\label{eq:def_fn}
    f_n(\widetilde x) \triangleq \rho \sum_{u=1}^{N_h}\sum_{v=1}^{N_v} p_{u,v} \frac{\eta \exp\!\big(-\beta r_{n,u,v}^2(\widetilde x)\big)+\mu_n^2}{r_{n,u,v}^2(\widetilde x)} .
\end{align}
\end{lem}

\emph{Proof:}
Starting from the definition of the traffic-weighted network average SNR and
substituting the closed-form per-grid average SNR in \eqref{eq:avg_snr_uv_closed_form}, we have
\begin{align}\label{eq:avg_obj_separable}
\bar{\Gamma}_{\mathrm{net}}(\widetilde{\boldsymbol{x}})
    &= \rho \sum_{u=1}^{N_h}\sum_{v=1}^{N_v} p_{u,v}
    \sum_{n=1}^{N}
    \frac{\eta \exp\!\big(-\beta r_{n,u,v}^2(\widetilde{x}_n)\big)+\mu_n^2}{r_{n,u,v}^2(\widetilde{x}_n)} \notag\\
    &\overset{(a)}{=}
    \sum_{n=1}^{N}
    \rho \sum_{u=1}^{N_h}\sum_{v=1}^{N_v} p_{u,v}\,
    \frac{\eta \exp\!\big(-\beta r_{n,u,v}^2(\widetilde{x}_n)\big)+\mu_n^2}{r_{n,u,v}^2(\widetilde{x}_n)} \notag\\
    &= \sum_{n=1}^{N} f_n(\widetilde{x}_n),
\end{align}
where step (a) follows by rearranging the finite sums (i.e., interchanging the
order of summation over $(u,v)$ and $n$). The proof is complete.
\hfill $\blacksquare$

Using Lemma~\ref{lem:net_traffic_problem_transform} and the separable box constraint set in \eqref{prob:P_Ave_box},
problem~\eqref{prob:P_Ave} decomposes into $N$ \emph{independent} scalar maximization problems:
\begin{align}\label{eq:avg_decomposed_1D}
    \max_{0\le x \le D_x} f_n(x),\quad \forall n\in\mathcal{N}.
\end{align}
Each $f_n(x)$ is continuously differentiable on $[0,D_x]$.
Nevertheless, $f_n(x)$ is generally nonconcave and may admit multiple local maximizers, since it is a traffic-weighted
superposition of many grid-wise terms that are individually peaked around different horizontal grid centers.

Despite this nonconcavity, \eqref{eq:avg_decomposed_1D} is substantially simpler than the original $N$-dimensional problem \eqref{prob:P_Ave}, as it involves only a single scalar variable with a compact feasible interval.
In principle, a globally optimal solution can be obtained by a one-dimensional search over $[0,D_x]$,
and the overall optimum of \eqref{prob:P_Ave} is then obtained by concatenating the $N$ independently optimized coordinates.
However, when $D_x$ is large and/or a fine search resolution is required, repeating such searches for all $N$ waveguides may still incur non-negligible complexity.
This motivates a dedicated solver that explicitly exploits the one-dimensional structure of
\eqref{eq:avg_decomposed_1D} while providing a principled way to handle the nonconcavity of $f_n(x)$.

\subsection{Problem Structure Analysis and Low-Complexity Algorithm Design}
In this section, we aim to develop low-complexity solution approaches to avoid
the prohibitive cost of exhaustive search. To this end, we start from the separable
per-waveguide formulation and further express each per-waveguide objective as a
traffic-weighted superposition of hotspot-induced components. We then
characterize the key structural properties of each component (e.g., smoothness
and unimodality). Building on these properties, we study the two-hotspot case in
detail to expose the fundamental merge--split behavior of the superposition,
which leads to a candidate-based low-complexity strategy for obtaining a high-quality solution by comparing only a small number of stationary and boundary candidates.
Finally, we extend the same principle to the general multi-hotspot case, yielding
scalable algorithms that remain effective as the number of hotspots grows.

To obtain clean structural insights, we work with the continuous-domain
definition of the traffic-weighted network average SNR, while the connection to the grid-based implementation will
be discussed in Remark~\ref{rem: continuous to discrete} at the end of this subsection.

Recall that, under the fully random-channel model adopted in this paper, the
traffic-weighted network average SNR can be written in a separable form as
\begin{align}
\bar{\Gamma}_{\mathrm{net}}(\widetilde{\boldsymbol{x}})
&= \sum_{n=1}^{N}\iint_{\mathcal{D}}
\bar{\Gamma}_n\!\big(x,y;\widetilde{x}_n\big)\,
f_{\boldsymbol{U}}(x,y)\, dx\,dy,
\label{eq:fn_def_continuous}
\end{align}
where $f_{\boldsymbol{U}}(x,y)$ is the spatial traffic density over the service
region $\mathcal{D}\subset\mathbb{R}^2$ given in \eqref{eq:pdf_hotspots}, and
$\bar{\Gamma}_n(x,y;\widetilde{x}_n)$ denotes the average SNR at location $(x,y,0)$
contributed by waveguide $n$ when its pinching antenna is placed at horizontal
coordinate $\widetilde{x}_n$. Specifically, the per-waveguide average SNR contribution is given by
\begin{equation}
    \bar{\Gamma}_n \big(x,y;\widetilde{x}_n\big) \triangleq \rho\,\frac{\eta\exp \!\big(-\beta r_n^2(x,y;\widetilde{x}_n)\big)+\mu_n^2}{r_n^2(x,y;\widetilde{x}_n)}.
    \label{eq:gamma_n_continuous}
\end{equation}

Next, substituting the Gaussian-mixture traffic model \eqref{eq:pdf_hotspots}
into \eqref{eq:fn_def_continuous} yields the hotspot decomposition
\begin{align}
\widetilde f_n(\widetilde{x}_n)
&= \iint_{\mathcal{D}} \bar{\Gamma}_n\!\big(x,y;\widetilde{x}_n\big)
\sum_{\ell=1}^{L} \alpha_{\ell}\,
\mathcal{N}\!\big([x,y]^{\mathsf T}; \boldsymbol{\mu}_{\ell},\boldsymbol{\Sigma}_{\ell}\big)\, dx\,dy \notag\\
&= \sum_{\ell=1}^{L} \alpha_{\ell}\, F_{n,\ell}(\widetilde{x}_n),
\label{eq:fn_gmm_decomp_cont}
\end{align}
where
\begin{equation}
F_{n,\ell}(\widetilde x_n)
\triangleq
\iint_{\mathcal{D}} \bar{\Gamma}_n\!\big(x,y;\widetilde x_n\big) 
\mathcal{N}\!\big([x,y]^{\mathsf T}; \boldsymbol{\mu}_{\ell},\boldsymbol{\Sigma}_{\ell}\big)\, dx\,dy.
\label{eq:Fnell_def_cont}
\end{equation}
Here, $F_{n,\ell}(\widetilde x_n)$ represents the contribution of hotspot $\ell$ to the
per-waveguide objective along waveguide $n$, and \eqref{eq:fn_gmm_decomp_cont}
shows that $\widetilde f_n(\cdot)$ is a traffic-weighted superposition of hotspot-induced
component functions. This representation will serve as the starting point for
the subsequent two-hotspot structural analysis.

\subsubsection{Shape of Each Hotspot-Induced Component $F_{n,\ell}(\widetilde x_n)$}
\label{subsubsec:shape_component}

We now study the structural properties of each hotspot-induced component
$F_{n,\ell}(\widetilde x_n)$ defined in \eqref{eq:Fnell_def_cont}. Throughout this
subsubsection, we fix a waveguide index $n$ and a hotspot index $\ell$, and view
$F_{n,\ell}(\cdot)$ as a scalar function of the pinching-antenna coordinate
$\widetilde x_n\in[0,D_x]$. For ease of analysis, we make the following assumption.

\begin{assumption} 
\label{ass:negligible_truncation}
The hotspot centers $\{\boldsymbol{\mu}_\ell\}_{\ell=1}^{L}$ lie well inside the
service region $\mathcal{D}=[0,D_x]\times[-D_y/2,D_y/2]$, and the traffic density
outside $\mathcal{D}$ is negligible.
\end{assumption}

Under Assumption~\ref{ass:negligible_truncation}, the key structural properties
of each hotspot-induced component $F_{n,\ell}(\widetilde x_n)$ are summarized in
the following lemma.

\begin{lem}
\label{lem:Fnell_unimodal}
Suppose that Assumption~\ref{ass:negligible_truncation} holds. Then, for any
waveguide $n$ and hotspot $\ell$, the function $F_{n,\ell}(\widetilde x_n)$
defined in \eqref{eq:Fnell_def_cont} is smooth and unimodal in $\widetilde x_n$
over $[0,D_x]$, and it attains its maximum at the hotspot horizontal center
$\widetilde x_n^\star=\mu_{\ell,x}$.
\end{lem}

\emph{Proof:} See Appendix~\ref{appd:F_symmetry}.
\hfill $\blacksquare$

\smallskip

\subsubsection{Two-Hotspot Case: Merge--Split Behavior}
\label{subsubsec:two_hotspot_merge_split}

We now specialize to the two-hotspot case ($L=2$) and study when the per-waveguide
objective
\begin{equation}
\widetilde f_n(x)=\alpha_1 F_{n,1}(x)+\alpha_2 F_{n,2}(x)
\label{eq:two_hotspot_fn_general}
\end{equation}
exhibits a \emph{single} dominant peak (``merged'' regime) or \emph{two} distinct
peaks (``split'' regime). This characterization is useful because it explains
why multimodality may arise even though each hotspot-induced component
$F_{n,\ell}(x)$ is unimodal (Lemma~\ref{lem:Fnell_unimodal}), and it enables
candidate-set construction that searches only a small number of informative
points.

Since $f_n(x)$ is differentiable, any interior maximizer $x^\star\in(0,D_x)$
must satisfy the first-order condition
\begin{equation}
\widetilde f_n'(x^\star)=\alpha_1 F_{n,1}'(x^\star)+\alpha_2 F_{n,2}'(x^\star)=0.
\label{eq:stationary_general}
\end{equation}
Intuitively, \eqref{eq:stationary_general} states that at an interior stationary
point the (traffic-weighted) pull from hotspot~1 (through $F_{n,1}'$) is
balanced by that from hotspot~2 (through $F_{n,2}'$).

Under Assumption~\ref{ass:negligible_truncation} and Lemma~\ref{lem:Fnell_unimodal},
each component $F_{n,\ell}(x)$ is smooth and unimodal over $[0,D_x]$ with a
unique maximizer at $x=\mu_{\ell,x}$. Hence, $F_{n,\ell}'(x)>0$ for
$x<\mu_{\ell,x}$ and $F_{n,\ell}'(x)<0$ for $x>\mu_{\ell,x}$. Without loss
of generality, assume $\mu_{1,x}<\mu_{2,x}$. Then, we have

\begin{small}
\begin{align}
\widetilde f_n'(\mu_{1,x})
&\!=\!\alpha_1 F_{n,1}'(\mu_{1,x})\!+\!\alpha_2 F_{n,2}'(\mu_{1,x}) \!=\! \alpha_2 F_{n,2}'(\mu_{1,x})\!>\! 0, \label{eq:deriv_at_mu1}\\
\widetilde f_n'(\mu_{2,x})
&\!=\!\alpha_1 F_{n,1}'(\mu_{2,x})\!+\!\alpha_2 F_{n,2}'(\mu_{2,x})\!=\!\alpha_1 F_{n,1}'(\mu_{2,x})\!<\!0. \label{eq:deriv_at_mu2}
\end{align}
\end{small}%
By continuity of $f_n'(x)$ and the intermediate value theorem \cite{courant1965introduction}, there exists
$x_b\in(\mu_{1,x},\mu_{2,x})$ such that $f_n'(x_b)=0$.

Although each $F_{n,\ell}$ is unimodal, the superposition $f_n$ may be
unimodal (merged regime) or bimodal (split regime). These behaviors are summarized in the following proposition.

\begin{prop}
\label{prop:two_hotspot_merge_split}
Assume that Assumption~\ref{ass:negligible_truncation} holds and
$\mu_{1,x}<\mu_{2,x}$. Then, $\widetilde f_n(x)$ has at least one stationary point
in $(\mu_{1,x},\mu_{2,x})$, i.e., there exists $x_b\in(\mu_{1,x},\mu_{2,x})$
such that $\widetilde f_n'(x_b)=0$. Moreover, depending on the number of stationary
points in $(\mu_{1,x},\mu_{2,x})$, $\widetilde f_n(x)$ exhibits the following regimes:
\begin{enumerate}
    \item (\emph{Merged regime}) If $\widetilde f_n$ has exactly one stationary point, denoted $x_b$, and $\widetilde f_n''(x_b)<0$, then $\widetilde f_n$ has
    a single dominant peak between the two hotspots; i.e., the maximizer of $\widetilde f_n$
    over $[\mu_{1,x},\mu_{2,x}]$ is unique and occurs at $x_b$.
    \item (\emph{Split regime}) If $\widetilde f_n$ has three stationary points in
    $(\mu_{1,x},\mu_{2,x})$, denoted $x_L<x_M<x_R$, such that
    \begin{equation}
    \widetilde f_n''(x_L)<0,\ \widetilde f_n''(x_M)>0,\ \widetilde f_n''(x_R)<0,
    \end{equation}
    then $\widetilde f_n$ has two distinct local maximizers $x_L$ and $x_R$ separated by a
    local minimizer $x_M$.
\end{enumerate}
\end{prop}

\emph{Proof:} See Appendix \ref{appd:prop_two_hotspot}. 
\hfill $\blacksquare$

\smallskip

Proposition~\ref{prop:two_hotspot_merge_split} suggests that, in the two-hotspot case,
it suffices to compare $f_n(x)$ over a small candidate set formed by interior
stationary points in $(\mu_{1,x},\mu_{2,x})$ and (when relevant) the endpoints.
A practical way to locate interior stationary points is to apply a
\emph{bracketing--bisection} procedure to the continuous derivative $h(x)\triangleq \widetilde f_n'(x)$.

\smallskip
\noindent\emph{Step 1 (Bracket sign changes).}
Since $h$ is continuous and satisfies $h(\mu_{1,x})>0$ and $h(\mu_{2,x})<0$,
there exists at least one root in $(\mu_{1,x},\mu_{2,x})$.
To identify sign-changing stationary-point candidates, we discretize
the interval into $J$ subintervals:
\[
x^{(0)}=\mu_{1,x}<x^{(1)}<\cdots<x^{(J)}=\mu_{2,x}.
\]
Evaluate $h(x^{(j)})$ for $j=0,\ldots,J$. For each adjacent pair
$[x^{(j)},x^{(j+1)}]$, if
\[
h(x^{(j)})\,h(x^{(j+1)})<0,
\]
then by the intermediate value theorem there exists at least one root
in $(x^{(j)},x^{(j+1)})$; this interval is a \emph{bracket} for a stationary
point of $f_n$.

\smallskip
\noindent\emph{Step 2 (Bisection within each bracket).}
For each bracket $[a,b]$ with $h(a)h(b)<0$, perform bisection:
set $c=(a+b)/2$; if $h(a)h(c)\le 0$ replace $b\leftarrow c$, else replace
$a\leftarrow c$. Repeat until $|b-a|\le \varepsilon$, and output
$x^\star\approx (a+b)/2$ as a stationary point.

\smallskip
\noindent\emph{Step 3 (Classify and build candidates).}
Compute $\widetilde f_n''(x^\star)$ (or use the sign change of $h$ around $x^\star$) to
classify the stationary point as a local maximizer/minimizer, and add the
maximizer-type points to the candidate set (merged: one point; split: two points).

\begin{rem} \label{rem: continuous to discrete}
The above structural analysis is established for the continuous-domain traffic model and is mainly used to reveal how different traffic hotspots jointly shape the per-waveguide objective. In practical implementation, we employ the grid-based approximation in \eqref{eq:avg_obj_separable}, which replaces the continuous traffic integral by a finite weighted sum and can be evaluated with substantially lower complexity. We emphasize that, for a finite grid resolution, the resulting grid-based objective is an approximation of its continuous-domain counterpart, and the unimodality and merge–split properties established above are therefore not guaranteed to hold exactly. Nevertheless, as the grid becomes sufficiently fine, the grid-based objective provides an increasingly accurate approximation of the continuous one. Motivated by the revealed continuous-domain structure, we apply a coarse-to-fine stationary-point search directly to the grid-based objective in the subsequent implementation.
\end{rem}

\subsubsection{Extension to the Multi-Hotspot Case}
\label{subsec:multi_hotspot_extension}

For $L>2$, the per-waveguide objective remains
\begin{equation}
f_n(x)=\sum_{\ell=1}^{L}\alpha_\ell F_{n,\ell}(x),
\end{equation}
and any interior maximizer $x^\star\in(0,D_x)$ must satisfy the equation
\begin{equation}
f_n'(x^\star)=\sum_{\ell=1}^{L}\alpha_\ell F'_{n,\ell}(x^\star)=0.
\label{eq:multi_stationary}
\end{equation}
Under Assumption~\ref{ass:negligible_truncation} and Lemma~\ref{lem:Fnell_unimodal},
each $F_{n,\ell}(x)$ is smooth and unimodal with a unique maximizer at
$x=\mu_{\ell,x}$, so $\widetilde f_n(x)$ may admit multiple stationary points and,
consequently, multi-peak profiles.

This observation motivates a scalable candidate-set construction: sort the
hotspot centers $\mu_{(1),x}\le\cdots\le\mu_{(L),x}$ and search for stationary
points of $f_n$ within each adjacent interval $(\mu_{(k),x},\mu_{(k+1),x})$,
$k=1,\ldots,L-1$, via coarse bracketing of sign changes in $\widetilde f_n'(x)$ followed by
bisection refinement. Collect the resulting stationary candidates and compare
$\widetilde f_n(x)$ only on this small set, instead of performing dense exhaustive search.
In the grid-based model, the same procedure is applied to the discrete surrogate
$f_n$ and its analytical derivative $f_n'$, so that bracketing and candidate comparison can be carried out
directly on the grid with low complexity.

\subsection{Extension to Discrete Pinching-Antenna Architectures}
The above formulation assumes continuously adjustable pinching-antenna
positions. In practical implementations, a finite set of pre-installed
pinching antennas may instead be deployed along each waveguide, with each
antenna operating in an activated/deactivated mode
\cite{tyrovolas2026many}. Let
$\mathcal{X}_n=\{x_{n,1},\ldots,x_{n,M}\}$ denote the candidate
locations on waveguide $n$, and let $a_{n,m}\in\{0,1\}$ indicate whether
candidate $m$ is activated, with $\sum_{m=1}^{M}a_{n,m}=1$.
Then, the traffic-weighted network average-SNR problem becomes
\begin{align}
\max_{\{a_{n,m}\}}\quad
&\sum_{n=1}^{N}\sum_{m=1}^{M}a_{n,m}f_n(x_{n,m})
\nonumber\\
\mathrm{s.t.}\quad
&\sum_{m=1}^{M}a_{n,m}=1,\quad
a_{n,m}\in\{0,1\},\ \forall n,m .
\end{align}
Importantly, the separable structure in Lemma \ref{lem:net_traffic_problem_transform} is preserved, and hence
the optimal activation on each waveguide is obtained independently
as
\begin{align}
m_n^\star=\arg\max_{m\in\{1,\ldots,M\}}f_n(x_{n,m}),\quad \forall n.
\end{align}
Therefore, the continuous one-dimensional optimization and stationary-point
search reduce to a finite comparison over the $M$ candidate locations,
with complexity linear in $M$ for each waveguide. This shows that the
proposed traffic-aware framework can be directly adapted to practical
two-state discrete pinching-antenna architectures.

\section{Traffic-Restricted Max--Min Average-SNR Optimization}
\label{sec:traffic_restricted_maxmin}

In this section, we consider a fairness-oriented design criterion that
aims to enhance the worst covered part of the network, while explicitly
taking into account the user traffic distribution.
Instead of enforcing coverage uniformly over the entire communication
area, including regions where users almost never appear, we restrict
our attention to a set of active grids with sufficiently large
traffic weights, and maximize the minimum local average SNR over this
set by optimizing the pinching-antenna positions.

\subsection{Active Grid Set and Problem Formulation}
\label{subsec:TRMM_formulation}

Recall from Section~\ref{subsec:grid_metrics} that the communication
region is discretized into grids indexed by
\[
  \Omega \triangleq \{1,\ldots,N_h\} \times \{1,\ldots,N_v\},
\]
with grid centers $\boldsymbol{\psi}_{u,v}$ and traffic weights
$\{p_{u,v}\}_{(u,v)\in\Omega}$ obtained from the Gaussian hotspot
model.
Given a traffic threshold $p_{\mathrm{th}} \in (0,1)$, we define the active grid set as
\begin{align}
  \Omega_{\mathrm{act}}
  \triangleq
    \big\{ (u,v)\in\Omega : p_{u,v} \ge p_{\mathrm{th}} \big\},
  \label{eq:Omega_act_def}
\end{align}
which contains the grids where users appear with non-negligible
probability.
Grids with $p_{u,v} < p_{\mathrm{th}}$ are excluded from the fairness
criterion, so that the design is not dominated by extremely low-traffic
regions.

We note that the threshold $p_{\rm th}$ controls the 
spatial scope over which fairness
is enforced. A smaller threshold includes more low-traffic grids and thus
provides fairness protection over a broader region, whereas a larger
threshold focuses the design on traffic-dominant regions. In practical
deployments, $p_{\rm th}$ can be selected according to the desired service
scope. For example, it can be chosen such that the resulting active set
covers a prescribed fraction of the total traffic probability. In this
work, we adopt the normalized form
$p_{\rm th}=\tau\max_{u,v}p_{u,v}$, which avoids dependence on the
absolute scale of the traffic map. The impact of the threshold selection
will be further examined in the simulation section.

The traffic-restricted worst-case average SNR is defined as
\begin{align}
\Gamma_{\min}^{\mathrm{(tr)}}(\widetilde{\boldsymbol{x}})
  \triangleq
    \min_{(u,v)\in \Omega_{\mathrm{act}}} \rho \sum_{n=1}^{N}
       \frac{
         \eta\, \exp\!\big(-\beta r_{n,u,v}^2(\widetilde{x}_n)\big)
         + \mu_n^2
       }{
         r_{n,u,v}^2(\widetilde{x}_n)
       },
  \label{eq:avg_snr_uv_TR}
\end{align}
where $\bar{\Gamma}_{u,v}(\widetilde{\boldsymbol{x}})$, given in \eqref{eq:avg_snr_uv_closed_form}, is the local average SNR in grid $(u,v)$ under deployment
$\widetilde{\boldsymbol{x}} = [\widetilde{x}_1,\ldots,\widetilde{x}_N]^{\mathsf T}$. 
$\Gamma_{\min}^{\mathrm{(tr)}}(\widetilde{\boldsymbol{x}})$ measures the coverage quality of the most disadvantaged active grid.

Based on the above model, the traffic-restricted max-min average SNR optimization problem is formulated as
\begin{subequations}
    \label{prob:P_MM_TR}
    \begin{align}
      \max_{\widetilde{\boldsymbol{x}}} \quad
      & \Gamma_{\min}^{\mathrm{(tr)}}(\widetilde{\boldsymbol{x}})\\
      \text{s.t.}\quad
      & 0 \le \widetilde{x}_n \le D_x,\quad \forall n \in \mathcal{N}.  
    \end{align}
\end{subequations}
Problem \eqref{prob:P_MM_TR} is, however, challenging to solve directly.
First, the objective is a max--min of highly nonlinear functions of
$\widetilde{\boldsymbol{x}}$, since each $\bar{\Gamma}_{u,v}(\widetilde{\boldsymbol{x}})$
in \eqref{eq:avg_snr_uv_TR} depends on the pinching-antenna positions
through the distances $r_{n,u,v}(\widetilde{x}_n)$ and the
LoS probabilities $\exp(-\beta r_{n,u,v}^2(\widetilde{x}_n))$,
leading to a highly nonconvex landscape with potentially many local optima.
Second, the outer minimum over $(u,v)\in\Omega_{\mathrm{act}}$
introduces additional nonsmoothness, so that improving the average SNR of some
grids may not immediately increase the worst-case value. Although global
exhaustive or grid-search-based methods are still possible in principle, their computational complexity grows explosively with $N_h$, $N_v$, and $N$ and quickly becomes prohibitive.

On the other hand, in the original formulation \eqref{prob:P_MM_TR} the
pinching-antenna positions are separable both in the average-SNR expressions and in the constraints.
This structure naturally motivates a BCD 
approach.
By updating one pinching-antenna position at a time and fixing the
others, the multi-dimensional max--min problem is reduced to a sequence
of one-dimensional subproblems, each with a simplified objective that
can be handled much more efficiently.

\subsection{BCD Method for the Original Max-Min Problem \eqref{prob:P_MM_TR}}
\label{subsec:TRMM_BCD_from_original}

To handle the nonconvex high-dimensional problem
\eqref{prob:P_MM_TR}, we adopt a BCD 
method.
At each iteration, we select one pinching antenna and adjust its
position along the waveguide while keeping all other antennas fixed.

Let $\widetilde{\boldsymbol{x}}^{(k)}$ denote the deployment at
iteration $k$.
From \eqref{eq:avg_snr_uv_TR}, the local average SNR in grid $(u,v)$
can be decomposed as
\begin{align}
  \bar{\Gamma}_{u,v}(\widetilde{\boldsymbol{x}})
  &= \rho \sum_{m=1}^{N}
       \frac{
         \eta\, \exp\!\big(-\beta r_{m,u,v}^2(\widetilde{x}_m)\big)
         + \mu_m^2
       }{
         r_{m,u,v}^2(\widetilde{x}_m)
       } \notag\\
  &= A_{u,v}^{(-n)} + g_{n,u,v}(\widetilde{x}_n),
  \quad (u,v)\in\Omega,
  \label{eq:TRMM_decomposition_original}
\end{align}
where
\begin{align}
  &A_{u,v}^{(-n)}
  \triangleq
    \rho \sum_{\substack{m=1\\ m\neq n}}^{N}
      \frac{
        \eta\, \exp\!\big(-\beta r_{m,u,v}^2(\widetilde{x}_m^{(k)})\big)
        + \mu_m^2
      }{
        r_{m,u,v}^2(\widetilde{x}_m^{(k)})
      },
    \label{eq:A_uv_minusn_TR_original}\\
  &g_{n,u,v}(\widetilde x_n)
  \triangleq
    \rho\,
    \frac{
      \eta\, \exp\!\big(-\beta r_{n,u,v}^2(\widetilde{x}_n)\big) + \mu_n^2
    }{
      r_{n,u,v}^2(\widetilde{x}_n)
    }.
    \label{eq:g_n_uv_TR_original}
\end{align}
Here, $A_{u,v}^{(-n)}$ collects the contribution of all antennas except
$n$ and is fixed during the update of $\widetilde{x}_n$, while
$g_{n,u,v}(\widetilde{x}_n)$ captures the contribution of antenna $n$ as a function
of its position $\widetilde{x}_n$.

Substituting \eqref{eq:TRMM_decomposition_original} into
\eqref{eq:avg_snr_uv_TR}, and fixing
$\{\widetilde{x}_m^{(k)}\}_{m\neq n}$, the original traffic-restricted
max--min objective becomes a one-dimensional function of $\widetilde{x}_n$, which is given by 
\begin{align}
  \Gamma_{\min}^{\mathrm{(tr)}}(\widetilde{x}_n)
  &= \min_{(u,v)\in\Omega_{\mathrm{act}}}
       \big( A_{u,v}^{(-n)} + g_{n,u,v}(\widetilde{x}_n) \big).
\end{align}
Therefore, the BCD update for antenna $n$ at iteration $k$ is given by
the one-dimensional max-min subproblem
\begin{align} 
  \max_{0 \le \widetilde{x}_n \le D_x} \quad
  & F_n^{\mathrm{(tr)}}(\widetilde{x}_n)
    \triangleq
      \min_{(u,v)\in\Omega_{\mathrm{act}}}
        \big( A_{u,v}^{(-n)} + g_{n,u,v}(\widetilde{x}_n) \big).
  \label{prob:P_MM_TR_n_original}
\end{align}
As seen, starting directly from the original max-min formulation
\eqref{prob:P_MM_TR}, the BCD scheme is naturally motivated.
For each antenna, we move its pinching position along the waveguide to
improve the worst local average SNR over all active grids, while the
contributions of the other antennas remain fixed in $A_{u,v}^{(-n)}$.
In what follows, we show that each one-dimensional subproblem can be
globally solved using a low-complexity bisection-based algorithm.

\subsection{Epigraph Reformulation and Bisection-Based Method for the One-Dimensional Subproblem \eqref{prob:P_MM_TR_n_original}} \label{subsec:TRMM_1D_bisection}

Although the coordinate subproblem \eqref{prob:P_MM_TR_n_original} is one-dimensional, it remains nonconvex due to the
max--min structure. Our key observation is that, under the smooth average-SNR model
\eqref{eq:TRMM_decomposition_original}--\eqref{eq:g_n_uv_TR_original}, the feasibility set induced by each grid constraint
has an interval structure, and the overall feasibility reduces to an interval intersection test.
This converts the global solution of the nonconvex one-dimensional max-min problem into a low-complexity nested bisection routine.

\smallskip
For a fixed antenna $n$ and given $\{A_{u,v}^{(-n)}\}$, the one-dimensional subproblem \eqref{prob:P_MM_TR_n_original}
admits the following epigraph reformulation:
\begin{subequations}\label{prob:P_MM_TR_n_epigraph_final}
\begin{align}
  \max_{0 \le \widetilde{x}_n \le D_x,\, t} \ 
  & t \label{prob:P_MM_TR_n_epigraph_final_a}\\
  \text{s.t.}\ 
  & A_{u,v}^{(-n)} + g_{n,u,v}(\widetilde{x}_n) \ge t,
    \  \forall (u,v)\in\Omega_{\mathrm{act}}, \label{prob:P_MM_TR_n_epigraph_final_b}
\end{align}
\end{subequations}
where $t$ represents the worst local average SNR over the active grids under the current update of antenna $n$.
While \eqref{prob:P_MM_TR_n_epigraph_final} is still nonconvex, its \emph{global} optimum can be obtained efficiently by exploiting
the monotone feasibility structure in $t$, as summarized below.

\begin{prop} 
\label{prop:TRMM_nested_bisection}
Fix antenna $n$ and $\{A_{u,v}^{(-n)}\}$. For each $(u,v)\in\Omega_{\mathrm{act}}$, define the per-grid feasibility set
\begin{align}
\mathcal{X}_{u,v}(t)
\triangleq
\Big\{
\widetilde{x}_n \in [0,D_x]
\,\Big|\,
A_{u,v}^{(-n)} + g_{n,u,v}(\widetilde{x}_n) \ge t
\Big\},
\label{eq:X_uv_t_TR}
\end{align}
and the overall feasibility set
\begin{align}
\mathcal{X}_n^{\mathrm{(tr)}}(t)
\triangleq
\bigcap_{(u,v)\in\Omega_{\mathrm{act}}} \mathcal{X}_{u,v}(t).
\label{eq:X_n_tr_t_final}
\end{align}
Then, the following properties hold
\begin{enumerate}
\item \emph{(Interval structure).}
For each $(u,v)\in\Omega_{\mathrm{act}}$, since
$g_{n,u,v}(x)$ is maximized at $x=x_u$ and decreases with
$|x-x_u|$, the feasible set $\mathcal{X}_{u,v}(t)$ is either empty
or a (possibly clipped) interval in $[0,D_x]$. When its boundary
lies within the feasible region, it can be written as
\begin{align}
\mathcal{X}_{u,v}(t)
=
[x_u-d_{u,v}(t),\,x_u+d_{u,v}(t)]\cap[0,D_x],
\end{align}
where $d_{u,v}(t)\geq 0$ is uniquely determined by
\begin{align} \label{eq:d_equation}
    A_{u,v}^{(-n)} +g_{n,u,v}\big(x_u+d_{u,v}(t)\big)=t. 
\end{align}
If the constraint is satisfied over the entire feasible region,
then $\mathcal{X}_{u,v}(t)=[0,D_x]$.
Consequently, $\mathcal{X}_{n}^{(\mathrm{tr})}(t)$ is either empty
or an interval in $[0,D_x]$.

\item \emph{(Nestedness).} For any $t_1\le t_2$, we have
$
\mathcal{X}_{u,v}(t_2)\subseteq \mathcal{X}_{u,v}(t_1)
$
for all $(u,v)\in\Omega_{\mathrm{act}}$, and hence
$
\mathcal{X}_n^{\mathrm{(tr)}}(t_2)\subseteq \mathcal{X}_n^{\mathrm{(tr)}}(t_1).
$

\item \emph{(Threshold optimality).} Define $t_n^\star \triangleq \sup\big\{\, t \,\big|\, \mathcal{X}_n^{\mathrm{(tr)}}(t)\neq\emptyset \,\big\}$.
Then, $\mathcal{X}_n^{\mathrm{(tr)}}(t)\neq\emptyset$ for all $t<t_n^\star$ and
$\mathcal{X}_n^{\mathrm{(tr)}}(t)=\emptyset$ for all $t>t_n^\star$.
Moreover, $t_n^\star$ equals the optimal objective value of \eqref{prob:P_MM_TR_n_epigraph_final}, and any
$\widetilde{x}_n^\star\in\mathcal{X}_n^{\mathrm{(tr)}}(t_n^\star)$ is an optimal solution.
\end{enumerate}
\end{prop}

The proof of Proposition \ref{prop:TRMM_nested_bisection} follows a similar argument to that of \cite[Proposition 1]{xu2025pinching-nlos} and is omitted here for brevity.
The properties established in Proposition \ref{prop:TRMM_nested_bisection} enable an efficient global solution method for \eqref{prob:P_MM_TR_n_epigraph_final}. In particular, the monotone nesting of $\mathcal{X}_n^{\mathrm{(tr)}}(t)$ allows $t_n^\star$ to be computed via a scalar bisection search over $t$. At each bisection iteration, we fix a candidate $t$, construct the
per-grid feasible intervals $\{\mathcal{X}_{u,v}(t)\}$, solving the
defining scalar equation for $d_{u,v}(t)$ when
$t>A_{u,v}^{(-n)}$, form their intersection
$\mathcal{X}_{n}^{(\mathrm{tr})}(t)$ according to~(49), and declare
$t$ feasible if $\mathcal{X}_{n}^{(\mathrm{tr})}(t)\neq\emptyset$. By iteratively shrinking the search range based on this feasibility test, we obtain $t_n^\star$ and a corresponding optimal $\widetilde{x}_n^\star$ with low complexity. This bisection-based solver is then used as the inner routine for each coordinate update in the overall BCD algorithm. We note that, although each one-dimensional coordinate subproblem is globally solved by the above bisection procedure, the overall BCD method does not in general guarantee global optimality for problem \eqref{prob:P_MM_TR}, and the
obtained solution may depend on the initialization and coordinate update
order.

To reduce the computational cost when $|\Omega_{\rm act}|$ is large, the feasibility test can be implemented with early termination. Specifically, the intervals ${\mathcal X_{u,v}(t)}$ are constructed and intersected sequentially, and the test terminates immediately once the running intersection becomes empty. In implementation, grids with larger SNR deficits under the current deployment can be examined first to facilitate early infeasibility detection. Moreover, since $g_{n,u,v}(x)$ is maximized at $x=x_u$, if
\begin{align}
A_{u,v}^{(-n)}+g_{n,u,v}(x_u)<t,
\end{align}
then $\mathcal X_{u,v}(t)=\emptyset$ and the tested $t$ can be declared infeasible without solving \eqref{eq:d_equation}. These operations only accelerate the feasibility check and do not change its result.

\begin{table}[!t]
\centering
\caption{Default simulation parameters.}
\renewcommand{\arraystretch}{1.25}
\begin{tabular}{p{4.8cm} p{3.2cm}}
\hline \hline
\textbf{Parameter} & \textbf{Value} \\
\hline
Carrier frequency ($f_c$) & $28$ GHz \\
Communication-region size ($D_x \times D_y$) & $60 \times 200$ m$^2$ \\
Waveguide height ($d_v$) & $10$ m \\
Transmit power ($P$) & $40$ dBm \\
Noise power ($\sigma^2$) & $-70$ dBm \\
NLoS power ($\mu^2$) &  $-60$ dBm \\
Number of PAs / waveguides ($N$) & $6$ \\
Number of grids ($N_h \times N_v$) & $400 \times 120$ \\
LoS blockage parameter ($\beta$) & $0.01$ \\
Algorithm tolerance ($\epsilon_t,\epsilon_d$) & $10^{-3}$ \\
Number of hotspots ($L$) & $3$ \\
Active-set threshold ($p_{\mathrm{th}}$) & $0.02 \times \max_{u,v} \{p_{u,v}\}$ \\
\hline \hline
\end{tabular}
\label{tab:sim-param}
\end{table}

Finally, we note that, for the traffic-restricted max--min design, the same two-state discrete architecture can be accommodated by restricting each coordinate update to the finite set of pre-installed pinching-antenna locations, such that each position update reduces to a direct comparison among the available candidates.

\section{Simulation Results}\label{sec:simulation} 
In this section, numerical simulations are presented to validate the performance of proposed algorithms for the traffic-aware pinching-antenna system designs. 
The main simulation parameters are listed in Table~\ref{tab:sim-param}. Unless otherwise stated, the default values in Table~\ref{tab:sim-param} are adopted. All curves are obtained by averaging over $100$ independent traffic realizations, each consisting of three randomly placed hotspots. The spatial traffic distribution is modeled as a three-component Gaussian mixture: for each realization, the hotspot centers are drawn uniformly from $[0,D_x]\times[-D_y/2,D_y/2]$, the mixture weights are independently generated and then normalized to sum to one, and all components share the same covariance matrix $\boldsymbol{\Sigma}=\mathrm{diag} \big((0.15D_x)^2,(0.2D_y)^2\big)$.

\begin{figure*}[!t]
	\begin{minipage}{0.32\linewidth}
		\centering
		\includegraphics[width=0.96\linewidth]{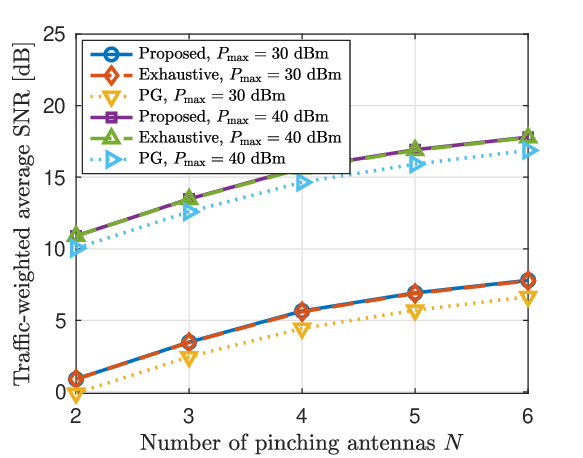}
		\captionsetup{justification=justified, singlelinecheck=false, font=small}	
		\caption{Traffic-weighted network average-SNR performance comparison among the proposed algorithm with two baselines.} \label{fig:tw_alg}  
	\end{minipage}~~
	\begin{minipage}{0.32\linewidth}
			\centering
				\includegraphics[width=0.96\linewidth]{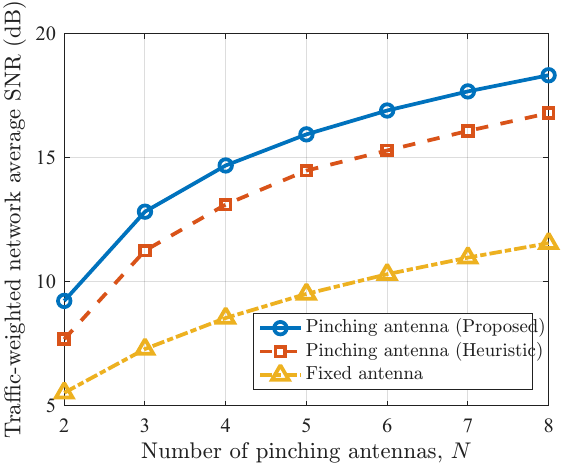}\\
        \captionsetup{justification=justified, singlelinecheck=false, font=small}	
        \caption{Traffic-aware network average SNR comparisons between different schemes versus the number of antennas $N$.} \label{fig:pg_snr_n_base}  
	\end{minipage}~~
    \begin{minipage}{0.32\linewidth}
		\centering
		\includegraphics[width=0.96\linewidth]{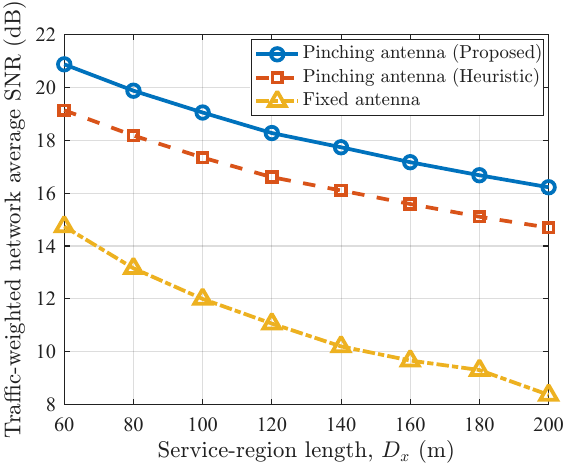}
		\captionsetup{justification=justified, singlelinecheck=false, font=small}	
		\caption{Traffic-aware network average-SNR comparisons between different schemes versus the communication-region length $D_x$.} \label{fig:pg_snr_dx_base}
	\end{minipage}~~
\end{figure*}

\subsection{Performance of the Traffic-Aware Network Average SNR Optimization}

\begin{table}[t] 
\centering
\caption{Execution time (in seconds) of the considered algorithms.}
\label{tab:time_exec}
\renewcommand\arraystretch{1.1}
\setlength{\tabcolsep}{0pt} 
\small

\newcolumntype{C}[1]{>{\centering\arraybackslash}p{#1}}

\begin{tabular}{C{16mm}|C{14mm}|C{14mm}|C{14mm}|C{14mm}|C{14mm}}
\hline
& \multicolumn{5}{c}{\cellcolor{blue!12} $\Delta x=\lambda/5,\; J=10, \; K_{\max} = 30$} \\
\cline{2-6}
& $N=2$ & $N=3$ & $N=4$ & $N=5$ & $N=6$ \\
\hline
Exhaustive & $7.83$ & $11.71$ & $15.6$ & $19.56$ & $23.43$ \\
Proposed   & $0.034$ &  $0.048$ & $0.059$ & $0.073$ &  $0.089$ \\
PG & $0.030$ & $0.044$ & $0.060$ & $0.070$ & $0.085$ \\
\hline
& \multicolumn{5}{c}{\cellcolor{blue!12} $\Delta x=\lambda/10,\; J=20,\; K_{\max} = 50$} \\
\cline{2-6}
& $N=2$ & $N=3$ & $N=4$ & $N=5$ & $N=6$ \\
\hline
Exhaustive & $15.61$ & $23.35$ & $30.88$ & $38.66$ & $46.61$ \\
Proposed   & $0.033$ & $0.049$ & $0.064$ & $0.077$ & $0.098$ \\
PG & $0.044$ & $0.072$ & $0.096$ & $0.130$ & $0.159$ \\
\hline
\end{tabular}
\end{table}

We first assess the performance of the proposed method by benchmarking it
against the exhaustive search baseline in Fig.~\ref{fig:tw_alg}, and also include
a projected gradient (PG) method as a representative local-optimization approach.
Exploiting the separable structure
$\bar{\Gamma}_{\mathrm{net}}=\sum_{n=1}^N f_n(\tilde x_n)$,
all three schemes reduce the design to $N$ independent one-dimensional maximization
problems.
For the exhaustive baseline, we conduct a dense linear search over $[0,D_x]$
with a uniform step size $\Delta_x=\lambda/5$ and choose the best grid point.
For the proposed method, we partition $[0,D_x]$ into $J=10$ coarse subintervals
to detect sign changes of $f_n'(x)$, refine each bracketed root via bisection,
and then evaluate $f_n(x)$ only at the resulting stationary-point candidates
(and endpoints when relevant). For PG, we iteratively update $\tilde x_n$ along
the gradient direction with step-size control and projection onto $[0,D_x]$ with maximum iterations of $K_{\max} = 30$.
 As seen in Fig.~\ref{fig:tw_alg}, the proposed method performs closely to the exhaustive
search for both $P=30$~dBm and $40$~dBm across all $N$, confirming that the
candidate-based search achieves a comparable objective value as exhaustive linear
search in this experiment while avoiding dense scanning. In contrast, PG yields
consistently lower values, since it is sensitive to initialization and may
converge to suboptimal stationary points in this nonconvex setting. Moreover,
the traffic-weighted average SNR increases with $N$ but exhibits diminishing
returns, and a larger $P$ consistently shifts the curves upward.

Table~\ref{tab:time_exec} further reports the average execution time of the three methods under two parameter settings: $\Delta x=\lambda/5,J=10,K_{\max}=30$ and $\Delta x=\lambda/10,J=20,K_{\max}=50$. In both settings, all methods exhibit an approximately increasing runtime with $N$, since they solve $N$ decoupled one-dimensional subproblems.
Several clear trends can be observed. First, the exhaustive linear search is the most time-consuming, because it must evaluate $f_n(x)$ over a dense grid on $[0,D_x]$ and the number of grid points scales inversely with $\Delta x$. Second, the proposed candidate-based method is consistently orders of magnitude faster than exhaustive search, since it avoids dense scanning and instead relies on coarse sign-change bracketing followed by only a small number of bisection refinements. Third, the PG method has comparable runtime to the proposed method in the first setting, but becomes noticeably slower in the second setting due to the larger iteration budget ($K_{\max}=50$). Overall, Table~\ref{tab:time_exec} confirms that the proposed approach offers a substantial computational advantage over exhaustive search, while remaining competitive with the PG method.

To quantify both the benefit of adopting pinching antennas and the extra gain from traffic-aware position optimization, Fig.~\ref{fig:pg_snr_n_base} compares the proposed design with two baselines. The first is a simple hotspot-center heuristic that places each pinching antenna at a hotspot's geometric center, representing coarse traffic-informed deployment. The second is a fixed-antenna benchmark that uses an $N$-element half-wavelength-spaced array centered at $x=D_x/2$ (with $\lambda/2$ spacing along $y$), representing a traffic-agnostic, non-reconfigurable design. As shown in Fig.~\ref{fig:pg_snr_n_base}, the proposed method achieves the highest network average SNR for all $N$, verifying the value of directly optimizing antenna locations under the traffic-weighted objective. The heuristic notably outperforms the fixed benchmark, but remains below the proposed design, while the fixed benchmark performs the worst due to its lack of adaptability to hotspot locations.



\begin{figure*}[!t]
    \begin{minipage}{0.32\linewidth}
		\centering
		\includegraphics[width=0.96\linewidth]{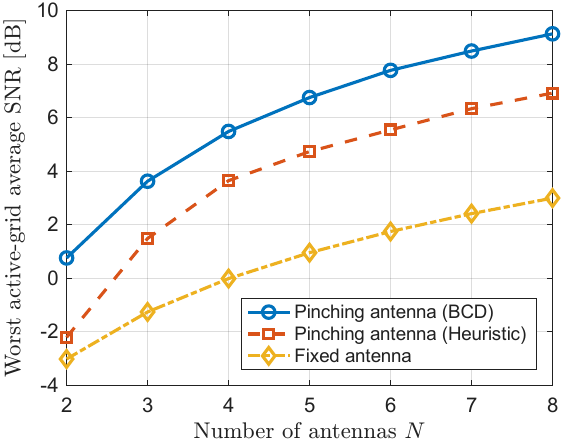}
		\captionsetup{justification=justified, singlelinecheck=false, font=small}	
		\caption{Worst active-grid average SNR comparisons with baselines versus the number of antennas $N$ for the traffic-restricted max-min design.} \label{fig:tr_snr_n} 
	\end{minipage}~~
	\begin{minipage}{0.32\linewidth} 
		\centering
		\includegraphics[width=0.96\linewidth]{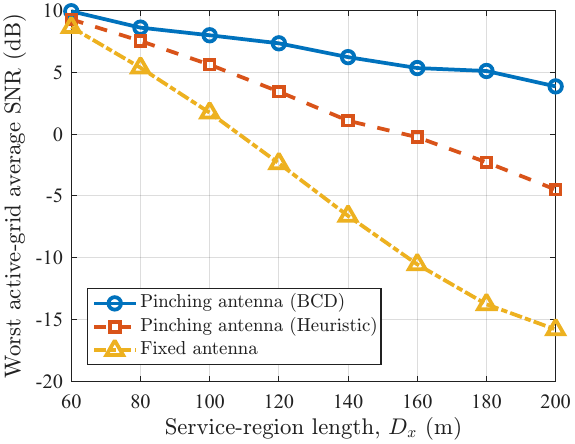}\\
		\captionsetup{justification=justified, singlelinecheck=false, font=small}	
        \caption{Worst active-grid average SNR comparisons with baselines versus the communication-region length $D_x$ for the traffic-restricted max-min design.}  \label{fig:tr_snr_dx} 
	\end{minipage}~~
	\begin{minipage}{0.32\linewidth}
			\centering
				\includegraphics[width=0.96\linewidth]{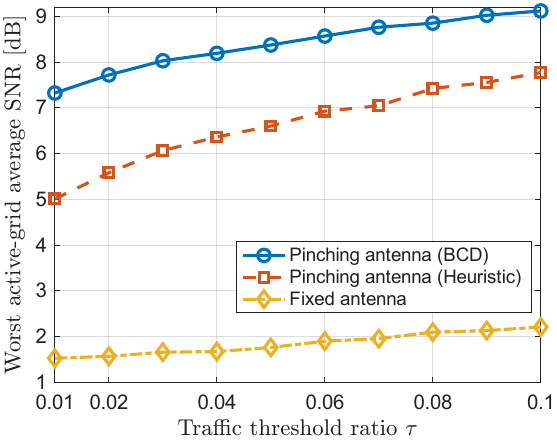}\\
        \captionsetup{justification=justified, singlelinecheck=false, font=small}	
        \caption{Worst active-grid average SNR comparisons with baselines versus the traffic threshold ratio $\tau$ for the traffic-restricted max-min design.} 
        \label{fig:tr_snr_tau} 
	\end{minipage}
\end{figure*}

To study how the gain of pinching antennas and the extra gain from traffic-aware position optimization scales with the service-region length, Fig.~\ref{fig:pg_snr_dx_base} compares the proposed design with the two baselines as $D_x$ varies. As $D_x$ increases, the network average SNR decreases for all schemes because the typical propagation distance grows, degrading link quality over the traffic field. Nevertheless, the proposed design consistently performs best across the entire $D_x$ range, demonstrating the robustness of traffic-aware optimization as the coverage area expands. The heuristic hotspot-center deployment remains competitive and clearly outperforms the fixed benchmark, confirming that even simple traffic-informed placement is beneficial, but it is still inferior to the proposed method. Moreover, the gap between the pinching-antenna designs and the fixed benchmark widens with $D_x$, indicating that placement adaptability becomes increasingly valuable in larger service regions.

\subsection{Performance of the Traffic-Restricted Max-Min Average SNR Optimization}

To illustrate the gain from pinching antennas and traffic-aware positioning, Fig.~\ref{fig:tr_snr_n} compares the proposed BCD-based design with two baselines (set following Fig.~\ref{fig:pg_snr_n_base}). The BCD solution consistently achieves the highest worst active-grid average SNR for all $N$, validating explicit optimization under the traffic-restricted max--min objective. The hotspot-center heuristic improves upon the fixed benchmark but remains inferior because it does not balance SNR across active grids. The fixed benchmark performs worst due to its traffic-agnostic, non-adaptive placement, which leaves a more severe bottleneck grid.
Fig.~\ref{fig:tr_snr_dx} studies how the pinching-antenna gain (and the extra gain from position optimization) changes with the service-region length $D_x$. The worst active-grid average SNR decreases with $D_x$ for all schemes due to larger typical propagation distances and a more challenging bottleneck in the active set. Still, the proposed BCD-based design consistently performs best over the entire $D_x$ range, showing its effectiveness in sustaining the worst-grid SNR as coverage expands. The hotspot-center heuristic degrades faster since it does not directly optimize the bottleneck, while the fixed benchmark suffers the largest loss, underscoring the poor robustness of traffic-agnostic, geometry-constrained deployments in large regions.

Fig.~\ref{fig:tr_snr_tau} evaluates how the traffic-restricted max-min design varies with the active-set threshold $p_{\mathrm{th}}=\tau\max_{u,v}p_{u,v}$, where $\tau$ controls how selectively traffic-dominant grids are included. As $\tau$ increases, the worst active-grid average SNR generally improves for all schemes because the active set becomes smaller and more concentrated on heavily loaded locations, easing the bottleneck. Throughout the entire $\tau$ range, the proposed BCD-based pinching-antenna design consistently achieves the highest worst active-grid SNR, showing that traffic-aware position optimization remains effective under different active-set definitions. The hotspot-center heuristic follows the same trend but stays noticeably below BCD since it does not explicitly balance the SNR across the active grids, while the fixed benchmark performs worst and is least robust to traffic-dependent selection.

To further evaluate the solution quality of the proposed BCD method, Table~\ref{tab:maxmin_global} compares it with a globally optimal branch-and-bound (BnB) benchmark in small-scale setups with $D_x=30$ m and $D_y=10$ m. To alleviate initialization sensitivity, multiple initial points are used for BCD, and the best converged solution is retained. The BnB solver is initialized with the BCD solution as a feasible incumbent, and its reported CPU time includes both the BCD initialization and the subsequent BnB search. As shown, BCD achieves an optimality gap of approximately $0.92$--$1.00$ dB for $N=3$--$5$. More importantly, while the BCD runtime remains around $0.11$ s, the runtime required to obtain the globally optimal solution increases rapidly with $N$, reaching $15.80$ s for $N=5$. These results demonstrate that the proposed BCD method provides a favorable tradeoff between solution quality and computational complexity.

\begin{table}[t]
\centering
\renewcommand\arraystretch{1.1}
\caption{Comparison with the globally optimal solution for the
traffic-restricted max--min average-SNR problem.}
\label{tab:maxmin_global}
\begin{tabular}{c|ccc|cc}
\hline
$N$
& \multicolumn{2}{c}{Worst active-grid SNR [dB]}
& Gap [dB]
& \multicolumn{2}{c}{CPU time [s]} \\
\cline{2-3}\cline{5-6}
& BCD & Global BnB & & BCD & Global BnB \\
\hline
3 & 20.792 & 21.710 & 0.918 & 0.111 & 0.169 \\
4 & 22.600 & 23.602 & 1.002 & 0.109 & 0.728 \\
5 & 23.821 & 24.817 & 0.996 & 0.114 & 15.801 \\
\hline
\end{tabular}
\end{table}

Fig.~\ref{fig:topview_obj_compare} illustrates how different traffic-aware objectives can yield different pinching-antenna placements under the same setting. We consider a single-waveguide scenario with $\widetilde y=0$ and one pinching antenna whose location $\widetilde x\in[0,D_x]$ is optimized. The background heatmap shows the normalized traffic weights $p_{u,v}$ from a four-hotspot Gaussian-mixture field. For this fixed traffic map, we compute (i) the network average-SNR maximization solution and (ii) the traffic-restricted max-min average-SNR solution, and mark both locations on the waveguide. As seen in Fig.~\ref{fig:topview_obj_compare}, the average-SNR design moves the antenna closer to the highest-traffic region to boost traffic-weighted efficiency, whereas the traffic-restricted max-min design shifts toward a more “balancing” position to improve the worst SNR within $\Omega_{\mathrm{act}}$. This contrast reflects the underlying philosophies: traffic-proportional averaging versus bottleneck-aware robustness within the active region.

\begin{figure}[!t]
	\centering
	\includegraphics[width=0.98\linewidth]{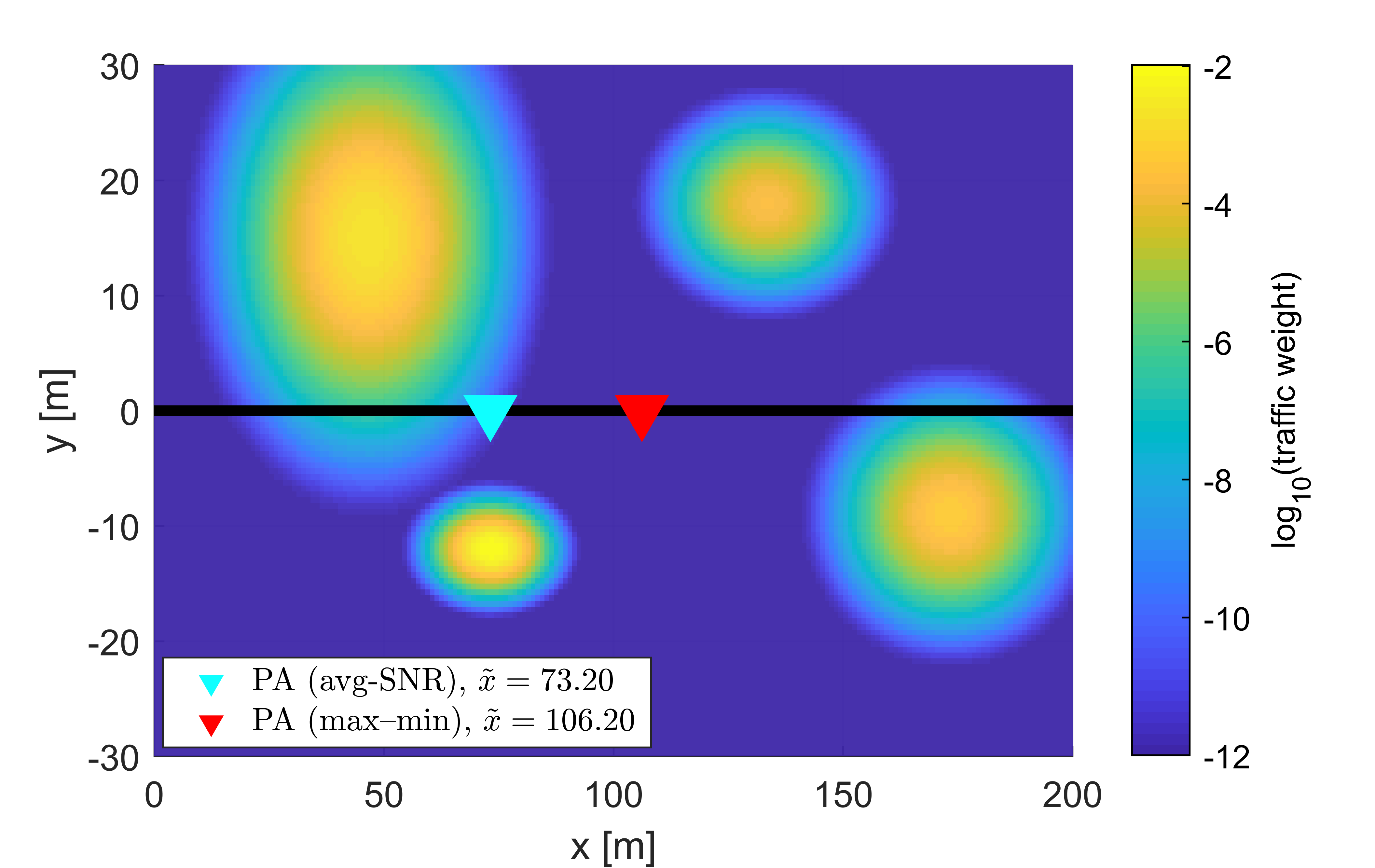}\\
        \captionsetup{justification=justified, singlelinecheck=false, font=small}	
        \caption{ Top-view illustration of the optimized pinching-antenna position under the network average-SNR maximization and the traffic-restricted max--min average-SNR maximization objectives for a fixed four-hotspot traffic field with a single waveguide at $\widetilde y=0$.} \label{fig:topview_obj_compare}  
\end{figure}

\subsection{Impact of In-Waveguide Attenuation}

We finally examine the impact of in-waveguide attenuation on the proposed
traffic-aware design. In practical pinching-antenna systems, the signal
experiences propagation loss while traveling along the waveguide before
being radiated, which may affect both the received signal strength and the
preferred pinching-antenna locations
\cite{tyrovolas2025performance}. Following the channel models in
\cite{xu2025pinching,xu2025losblockage}, the channel from the $n$th
waveguide can be extended as
\begin{align}
h_n^{\mathrm{att}}(\boldsymbol{\psi},\tilde{\mathbf{x}})
=
e^{-\alpha \tilde{x}_n}
\left(
\xi_n h_n^{\mathrm{LoS}}(\boldsymbol{\psi},\tilde{\mathbf{x}})
+
h_n^{\mathrm{NLoS}}(\boldsymbol{\psi},\tilde{\mathbf{x}})
\right),
\label{eq:att_total_channel}
\end{align}
where $\alpha$ denotes the in-waveguide amplitude attenuation coefficient
in $\mathrm{m}^{-1}$, and $h_n^{\mathrm{LoS}}$ and
$h_n^{\mathrm{NLoS}}$ follow the definitions in Sec.~II-B.
Accordingly, the average received SNR at grid $(u,v)$ becomes
\begin{align}
\bar{\Gamma}_{u,v}^{\mathrm{att}}(\tilde{\mathbf{x}})
=
\rho
\sum_{n=1}^{N}
e^{-2\alpha \tilde{x}_n}
\frac{
\eta \exp\!\left(-\beta r_{n,u,v}^2(\tilde{x}_n)\right)
+\mu_n^2
}{
r_{n,u,v}^2(\tilde{x}_n)
},
\label{eq:att_average_snr}
\end{align}
and the corresponding traffic-weighted network average SNR is
\begin{align}
\bar{\Gamma}_{\mathrm{net}}^{\mathrm{att}}(\tilde{\mathbf{x}})
=
\sum_{u=1}^{N_h}\sum_{v=1}^{N_v}
p_{u,v}
\bar{\Gamma}_{u,v}^{\mathrm{att}}(\tilde{\mathbf{x}}).
\label{eq:att_network_snr}
\end{align}

To quantify the impact of neglecting in-waveguide attenuation in the
position design, we compare two schemes. The \emph{attenuation-aware}
design optimizes the pinching-antenna positions according to
\eqref{eq:att_network_snr}, whereas the \emph{attenuation-ignorant}
design optimizes the positions using the original attenuation-free model
and subsequently evaluates the resulting deployment under the same
attenuated channel. Hence, the performance difference between the two
schemes directly characterizes the additional loss caused by neglecting
in-waveguide attenuation during position optimization. To isolate this
modeling effect, both designs are numerically optimized using the same
dense one-dimensional position search.

\begin{table}[t]
\centering 
\renewcommand{\arraystretch}{1.1}
\setlength{\tabcolsep}{7.5pt}
\caption{Impact of in-waveguide attenuation on the traffic-aware
pinching-antenna position design.}
\label{tab:attenuation}
\footnotesize
\begin{tabular}{c|cc|c|c}
\hline
$\alpha_{\mathrm{dB}}$
& \multicolumn{2}{c|}{Network average SNR [dB]}
& Loss [dB]
& Shift [m] \\
{}[dB/m] & Aware & Ignorant & & \\
\hline
0    & 20.952 & 20.952 & 0.000 & 0.000 \\
0.01 & 20.806 & 20.805 & 0.001 & 0.156 \\
0.02 & 20.661 & 20.659 & 0.003 & 0.312 \\
0.04 & 20.379 & 20.367 & 0.012 & 0.724 \\
0.06 & 20.105 & 20.077 & 0.028 & 1.056 \\
0.08 & 19.840 & 19.788 & 0.053 & 1.526 \\
0.10 & 19.586 & 19.500 & 0.086 & 1.891 \\
\hline
\end{tabular}
\end{table}

For numerical evaluation, we consider $D_x=100$ m, $D_y=30$ m, $N=6$, and
$N_h\times N_v=200\times60$, while the remaining parameters follow
Table \ref{tab:sim-param}. We vary the propagation loss as
$\alpha_{\mathrm{dB}}\in
\{0,0.01,0.02,0.04,0.06,0.08,0.1\}$ dB/m.
In particular, the range is consistent with the
experimentally reported propagation losses of dielectric
waveguides at millimeter- and submillimeter-wave frequencies
\cite{yeh2000communication}.
Table~\ref{tab:attenuation} summarizes the results, where ``Loss'' denotes
the traffic-aware network average-SNR loss of the attenuation-ignorant design relative
to the attenuation-aware design, and ``Shift'' denotes the average
displacement of the attenuation-aware pinching-antenna positions toward
the feed points relative to their attenuation-ignorant counterparts.
As expected, increasing the in-waveguide attenuation reduces the network
average SNR and gradually shifts the optimized radiation locations toward
the feed points, consistent with the impact of waveguide attenuation
reported in \cite{tyrovolas2025performance}. For example, at
$\alpha_{\mathrm{dB}}=0.08$ dB/m, the average feed-side displacement is
about $1.53$ m. Nevertheless, the additional SNR loss
caused by neglecting attenuation during position optimization remains no
more than $0.053$ dB. Even when $\alpha_{\mathrm{dB}}$ is increased to
$0.1$ dB/m, this loss is only $0.086$ dB. These results indicate that,
although in-waveguide attenuation reduces the absolute received SNR and
shifts the preferred pinching-antenna locations toward the feed points,
neglecting it in the position optimization introduces only a minor
additional performance loss for the considered setting. 
For longer waveguides or guiding media with higher
propagation loss, however, explicitly accounting for in-waveguide
attenuation may become more important.

\begin{figure}[!t]
	\centering
	\includegraphics[width=0.82\linewidth]{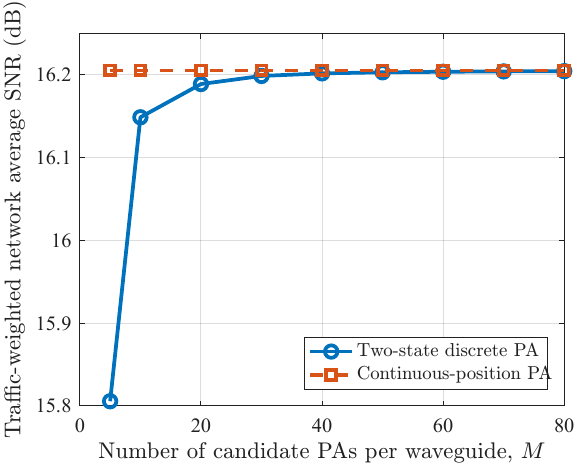}\\
        \captionsetup{justification=justified, singlelinecheck=false, font=small}	
        \caption{Traffic-weighted network average-SNR comparison between continuous-position and two-state discrete pinching-antenna implementations versus the number of candidate pinching antennas per waveguide.} \label{fig:two_state}  
\end{figure}

\subsection{Continuous Implementation Versus Discrete Implementation}
Fig.~\ref{fig:two_state} compares the continuously adjustable
pinching-antenna architecture with a two-state discrete implementation,
where $M$ pre-installed candidate pinching antennas are uniformly deployed
along each waveguide and one candidate is activated at a time. For the
two-state implementation, the optimal candidate on each waveguide is
independently selected according to the discrete counterpart of the
traffic-weighted network average-SNR objective. As shown, increasing $M$
rapidly reduces the performance gap relative to the continuous-position
design. When only a small number of candidate locations are available,
a noticeable discretization loss is observed. As $M$ increases, however,
the two-state design rapidly approaches the continuous-position benchmark
and becomes nearly indistinguishable from it with a moderate number of
candidate locations. These results demonstrate that the proposed
traffic-aware framework can be effectively implemented with practical
two-state pinching-antenna architectures without requiring continuously
adjustable radiation locations.

\section{Conclusion} \label{sec: conclusion}
In this paper (Part I), we investigated traffic-aware network-level design for pinching-antenna systems. 
Complementing existing link-level pinching-antenna optimization, we introduced a traffic-aware network-level framework based on long-term spatial traffic information.
A grid-based discretization was further developed to obtain numerically efficient discrete counterparts of the continuous-domain metrics. Building on these models, we formulated two traffic-aware deployment problems: maximizing the traffic-weighted network average SNR and a fairness-oriented traffic-restricted max–min average-SNR design over traffic-dominant grids. Efficient low-complexity algorithms were proposed, including a candidate-based low-complexity optimization algorithm for the traffic-weighted objective and a BCD method with bisection-based coordinate updates for the max–min objective. Numerical results demonstrated the effectiveness of the proposed algorithms and showed that the proposed traffic-aware designs can outperform representative fixed and heuristic baselines under the considered scenarios, while also revealing how the achievable gains and robustness–efficiency tradeoffs depend on the traffic heterogeneity and system parameters.

\begin{appendices}
    \section{Proof of Lemma \ref{lem:average_snr}} \label{appd: average snr} 
    
According to the random LoS/NLoS channel model, the channel coefficient of $n$-th pinching antenna is given by
\begin{align}
  h_n(\boldsymbol{\psi}_{u,v},\widetilde{\boldsymbol{x}})
  = \xi_n h^{\mathrm{LoS}}_n(\boldsymbol{\psi}_{u,v},\widetilde{\boldsymbol{x}})
    + h^{\mathrm{NLoS}}_n(\boldsymbol{\psi}_{u,v},\widetilde{\boldsymbol{x}}),
\end{align}
where $\xi_n \in \{0,1\}$ is the LoS indicator, and the LoS
probability is given by
\begin{align}
  \Pr\big[\xi_n = 1\big] = \exp\!\big(-\beta r_{n,u,v}^2(\widetilde{\boldsymbol{x}})\big).
\end{align}

From the LoS expression in Section~\ref{subsec:channel_model}, we have
\begin{align}
  \big|h^{\mathrm{LoS}}_n(\boldsymbol{\psi}_{u,v},\widetilde{\boldsymbol{x}})\big|^2
  = \frac{\eta}{r_{n,u,v}^2(\widetilde{\boldsymbol{x}})}.
\end{align}
For the NLoS component, we have
\begin{align}
  \mathbb{E}\Big[
    \big|h^{\mathrm{NLoS}}_n(\boldsymbol{\psi}_{u,v},\widetilde{\boldsymbol{x}})\big|^2
  \Big]
  = \frac{\mu_n^2}{r_{n,u,v}^2(\widetilde{\boldsymbol{x}})}.
\end{align}

Using the independence between $\xi_n$ and
$h^{\mathrm{NLoS}}_n$, and the zero mean of the NLoS term, the average
channel power contributed by the $n$-th pinching antenna to grid
$(u,v)$ gives

\begin{small}
\begin{align}
  \bar{G}_{n,u,v}(\widetilde{\boldsymbol{x}})
  &\triangleq
    \mathbb{E}\Big[
      \big|h_n(\boldsymbol{\psi}_{u,v},\widetilde{\boldsymbol{x}})\big|^2
    \Big] \notag\\
  &= \mathbb{E}[\xi_n]\,
     \big|h^{\mathrm{LoS}}_n(\boldsymbol{\psi}_{u,v},\widetilde{\boldsymbol{x}})\big|^2
     + \mathbb{E}\Big[
         \big|h^{\mathrm{NLoS}}_n(\boldsymbol{\psi}_{u,v},\widetilde{\boldsymbol{x}})\big|^2
       \Big] \notag\\
  &= \frac{
       \eta\, \exp\!\big(-\beta r_{n,u,v}^2(\widetilde{\boldsymbol{x}})\big)
       + \mu_n^2
     }{
       r_{n,u,v}^2(\widetilde{\boldsymbol{x}})
     }.
  \label{eq:avg_channel_power_grid_n}
\end{align}
\end{small}%
Therefore, the local average SNR in grid $(u,v)$ is
\begin{small}
\begin{align}
  \bar{\Gamma}_{u,v}(\widetilde{\boldsymbol{x}})
  &\triangleq
    \mathbb{E}\big[
      \Gamma_{u,v}(\widetilde{\boldsymbol{x}})
    \big] \notag\\
  &= \rho \sum_{n=1}^{N}
       \bar{G}_{n,u,v}(\widetilde{\boldsymbol{x}}) \notag\\
  &= \rho \sum_{n=1}^{N}
       \frac{ \eta\, \exp\!\big(-\beta r_{n,u,v}^2(\widetilde{\boldsymbol{x}})\big) + \mu_n^2
       }{ r_{n,u,v}^2(\widetilde{\boldsymbol{x}}) }.
\end{align}
\end{small}

Using \eqref{eq:net_avg_SNR_grid}, the traffic-aware
network average SNR can now be written as \eqref{eq:net_avg_SNR_grid_final}.
The proof is complete.  \hfill $\blacksquare$

\section{Proof of Lemma \ref{lem:Fnell_unimodal}} \label{appd:F_symmetry}

First, let $(X_\ell,Y_\ell)\sim\mathcal{N}(\boldsymbol{\mu}_\ell,\boldsymbol{\Sigma}_\ell)$.
Under Assumption \ref{ass:negligible_truncation} and by the definition of expectation under a PDF, \eqref{eq:Fnell_def_cont} can be written as
\begin{equation}
    F_{n,\ell}(\widetilde x_n)
    = \mathbb{E}_{X_\ell,Y_\ell}\!\left[\bar{\Gamma}_n(X_\ell,Y_\ell;\widetilde x_n)\right].
    \label{eq:Fnell_expectation}
\end{equation}
Then, by defining the scalar mapping $k_n(s)\triangleq \rho\,\frac{\eta e^{-\beta s}+\mu_n^2}{s},\, s>0$ and using \eqref{eq:gamma_n_continuous}, we can rewrite the per-waveguide average SNR as $\bar{\Gamma}_n(x,y;\widetilde x_n)=k_n\!\big(r_n^2(x,y;\widetilde x_n)\big)$.
Moreover, $k_n(s)$ is strictly decreasing on $s>0$ because both
$e^{-\beta s}/s$ and $1/s$ are strictly decreasing for $s>0$.
Since $\boldsymbol{\Sigma}_\ell$ is diagonal, $X_\ell$ and $Y_\ell$ are
independent. Fix $t\ge 0$ and define
\begin{equation}
\phi_{n,\ell}(t)
\triangleq
\mathbb{E}_{Y_\ell}\!\left[
k_n\!\big(t+(\widetilde y_n - Y_\ell)^2+d_v^2\big)
\right],
\label{eq:phi_nl_def}
\end{equation}
where the expectation is taken with respect to
$Y_\ell\sim\mathcal{N}(\mu_{\ell,y},\sigma_{\ell,y}^2)$.
Since $k_n(\cdot)$ is strictly decreasing, $\phi_{n,\ell}(t)$ is also strictly
decreasing on $t\ge 0$. Now, applying the law of iterated expectation to \eqref{eq:Fnell_expectation}
gives
\begin{align}
F_{n,\ell}(\widetilde x_n)
&=\mathbb{E}_{X_\ell}\!\Big[
\mathbb{E}_{Y_\ell}\!\big[
k_n\!\big((\widetilde x_n-X_\ell)^2+(\widetilde y_n - Y_\ell)^2+d_v^2\big)
\big]
\Big] \notag\\
&=\mathbb{E}_{X_\ell}\!\left[\phi_{n,\ell}\big((\widetilde x_n-X_\ell)^2\big)\right].
\label{eq:Fnell_iterated}
\end{align}
Define $g_{n,\ell}(u)\triangleq \phi_{n,\ell}(u^2),\, u\in\mathbb{R}$. 
Then $g_{n,\ell}(u)$ is even (i.e., $g_{n,\ell}(u)=g_{n,\ell}(-u)$) and strictly
decreasing in $|u|$ because $\phi_{n,\ell}(t)$ is strictly decreasing in $t$.
Let $f_{X_\ell}$ denote the PDF of
$X_\ell\sim\mathcal{N}(\mu_{\ell,x},\sigma_{\ell,x}^2)$. Writing the expectation
in \eqref{eq:Fnell_iterated} as an integral yields
\begin{align}
F_{n,\ell}(\widetilde x_n)
&=\int_{-\infty}^{\infty} \phi_{n,\ell}\big((\widetilde x_n-\xi)^2\big)\, f_{X_\ell}(\xi)\, d\xi \notag\\
&=\int_{-\infty}^{\infty} g_{n,\ell}(\widetilde x_n-\xi)\, f_{X_\ell}(\xi)\, d\xi \notag\\
&=(g_{n,\ell}*f_{X_\ell})(\widetilde x_n),
\label{eq:Fnell_convolution}
\end{align}
where $(g*f)(x)\triangleq\int g(x-\xi)f(\xi)\,d\xi$ denotes convolution.

Since $g_{n,\ell}$ is symmetric unimodal and $f_{X_\ell}$ is Gaussian (thus symmetric unimodal),
their convolution is unimodal (e.g., \cite[Theorem~2.1]{purkayastha1998simple}).
We next show it is symmetric about $\mu_{\ell,x}$.
Letting $\widetilde x_n=\mu_{\ell,x}+z$ and changing variables $\xi=\mu_{\ell,x}+t$, we have

\begin{small}
\begin{align}
F_{n,\ell}(\mu_{\ell,x}+z)
&=\int_{-\infty}^{\infty} g_{n,\ell}(z-t)\, f_0(t)\,dt \notag\\
&\overset{(a)}{=}\int_{-\infty}^{\infty} g_{n,\ell}\big(-(z-t)\big)\, f_0(t)\,dt \notag\\
&\overset{(b)}{=}\int_{-\infty}^{\infty} g_{n,\ell}(-z+t)\, f_0(t)\,dt \notag\\
&\overset{(c)}{=}\int_{-\infty}^{\infty} g_{n,\ell}(-z-t)\, f_0(-t)\,dt \notag\\
&\overset{(d)}{=}\int_{-\infty}^{\infty} g_{n,\ell}(-z-t)\, f_0(t)\,dt \notag\\
&=F_{n,\ell}(\mu_{\ell,x}-z),
\end{align}
\end{small}
where $f_0(t)\triangleq f_{X_\ell}(\mu_{\ell,x}+t)$ is even and $g_{n,\ell}$ is even,
(a) uses the evenness of $g_{n,\ell}(\cdot)$, i.e., $g_{n,\ell}(u)=g_{n,\ell}(-u)$,
(b) uses $-(z-t)=-z+t$, (c) applies the change of variable $t\mapsto -t$, and
(d) uses the evenness of $f_0(\cdot)$, i.e., $f_0(t)=f_0(-t)$.
Hence $F_{n,\ell}$ is symmetric about $\mu_{\ell,x}$ and, being unimodal,
achieves its maximum at $\widetilde x_n=\mu_{\ell,x}$. The proof is complete.
\hfill $\blacksquare$

\section{Proof of Proposition \ref{prop:two_hotspot_merge_split}} \label{appd:prop_two_hotspot}
The existence of the stationary point follows from \eqref{eq:deriv_at_mu1}--\eqref{eq:deriv_at_mu2} and the intermediate value theorem applied to the continuous function $\widetilde f_n'(x)$ on
$[\mu_{1,x},\mu_{2,x}]$.

For merge regime, if $x_b$ is the unique stationary point in $(\mu_{1,x},\mu_{2,x})$
and satisfies $\widetilde f_n''(x_b)<0$, then $x_b$ is a strict local maximizer. Moreover,
since $\widetilde f_n'(\mu_{1,x})>0$ and $\widetilde f_n'(\mu_{2,x})<0$, the derivative must change
sign from positive to negative across $(\mu_{1,x},\mu_{2,x})$. Uniqueness of the
stationary point forces this sign change to occur at $x_b$, implying that $\widetilde f_n$
is increasing on $[\mu_{1,x},x_b]$ and decreasing on $[x_b,\mu_{2,x}]$. Hence,
$x_b$ is the unique maximizer of $\widetilde f_n$ over $[\mu_{1,x},\mu_{2,x}]$.

For split regime, the second-order conditions imply that $x_L$ and $x_R$ are strict
local maximizers while $x_M$ is a strict local minimizer. Since
$x_L<x_M<x_R$, the two maximizers are distinct and are separated by a valley at
$x_M$, yielding a two-peak (split) structure on $(\mu_{1,x},\mu_{2,x})$.
\hfill $\blacksquare$

\end{appendices}



\end{document}